\documentclass[11pt,a4paper]{article}
\usepackage{amssymb,amsmath,amsthm}
\usepackage{authblk}
\usepackage{multirow}
\usepackage{blindtext}
\usepackage{amsfonts}
\usepackage{amsmath,amssymb,pxfonts, xcolor}
\usepackage{graphicx}
\usepackage{ulem}

\usepackage{color}
\usepackage{soul}
\usepackage{float}
\usepackage{comment}
\usepackage{adjustbox}

\newtheorem{theorem}{Theorem}

\newtheorem{corollary}[theorem]{Corollary}

\newtheorem{proposition}[theorem]{Proposition}
\newtheorem{remark}[theorem]{Remark}

\numberwithin{equation}{section}

\makeatletter
\newcommand{\biggg}[1]{{\hbox{$\left#1\vbox to 20.5pt{}\right.\n@space$}}}
\newcommand{\Biggg}[1]{{\hbox{$\left#1\vbox to 23.5pt{}\right.\n@space$}}}
\newcommand{\bigggg}[1]{{\hbox{$\left#1\vbox to 26.5pt{}\right.\n@space$}}}
\newcommand{\Bigggg}[1]{{\hbox{$\left#1\vbox to 29.5pt{}\right.\n@space$}}}
\newcommand{\biggggg}[1]{{\hbox{$\left#1\vbox to 32.5pt{}\right.\n@space$}}}
\newcommand{\Biggggg}[1]{{\hbox{$\left#1\vbox to 35.5pt{}\right.\n@space$}}}
\newcommand{\bigggggg}[1]{{\hbox{$\left#1\vbox to 38.5pt{}\right.\n@space$}}}
\newcommand{\Bigggggg}[1]{{\hbox{$\left#1\vbox to 41.5pt{}\right.\n@space$}}}

\title{Forward start volatility swaps in rough volatility models}
\author{Elisa Al\`{o}s\thanks{Dpt. d'Economia i Empresa, Universitat Pompeu Fabra.} \quad Frido Rolloos \quad Kenichiro Shiraya\thanks{Graduate School of Economics, The University of Tokyo. Kenichiro Shiraya is supported by CARF.}}

\begin{document}
\maketitle

\begin{abstract}
This paper shows the relationship between the forward start volatility swap price and the forward start zero vanna implied volatility of forward start options in rough volatility models. It is shown that in the short time-to-maturity limit the approximation error in the leading term of the correlated case with $H\in(0,\frac12)$ does  not depend on the time to forward start date, but only on the difference between the maturity date and forward start date and on the Hurst parameter $H$. 
\bigskip

Keywords: Rough volatiliy, volatility swap, implied volatility, Malliavin calculus

AMS subject classification: 91G99
\end{abstract}

\section{Introduction}
In a Black-Scholes (BS) setting, the valuation of forward starting options is as straightforward as the pricing of vanilla options. A simple application of conditional expectation shows that in a BS world the value of a forward start option is in fact equal to the price of a vanilla option with a time to maturity that does not vary until the forward start date (also known as the strike date). This result was first derived by Rubinstein \cite{MR}.

When the BS assumptions are relaxed, in particular when the instantaneous volatility is driven by fractional noise, the problem of pricing forward starting derivatives becomes substantially more complex and computationally expensive. Even for forward start call and put options, unless the characteristic function of the model is known in closed-form, an exact construction of the forward start smile will require numerical schemes such as Monte Carlo simulation. The Heston model, for which the characteristic function is known in closed form, has therefore been extensively studied in regard to forward starting derivatives (see for instance \cite{KN}). 

Nevertheless, general asymptotic features of the forward start smile have been analysed in several papers and it is possible to draw some conclusions without reference to a specific model. For instance, under general stochastic volatility models Al\`os et al. \cite{AJV} show that, in contrast to the vanilla short time-to-maturity ATM level, the forward start short time-to-maturity ATM level of Type II forward start options is a direct function of the correlation between the underlying and its instantaneous volatility. They also prove that the forward start ATM skew decays at a different rate than the vanilla ATM skew. In general though, it can be said that the literature and results on the pricing and hedging of forward starting derivatives is less extensive than for spot starting products. This is especially true for forward start volatility derivatives. 

The most commonly traded forward start volatility derivatives are forward start variance swaps and VIX futures. The former can be synthesised from vanilla options only by making use of the additive property of variance. VIX futures on the other hand cannot be synthesised from vanilla options and are therefore dependent on the specific model describing future evolution of the underlying asset and its volatility. 

Like VIX futures, the forward start volatility swap also depends on the assumptions of forward dynamics and is model-dependent. From Jensen's inequality it follows that the price of the VIX future is bounded below by the price of the forward start volatility swap. The literature on forward start volatility swaps is even more sparse than papers treating pricing of VIX futures. However, as will be shown, it is possible to give an analytical approximation for the forward start volatility swap that can be read off the forward start implied volatility smile directly and which is valid for a wide class of models. This result is an extension of \cite{ARS}. In this paper we will quantify rigorously the difference between the approximation and the exact price. 

Our paper is structured as follows. In Section 2 the problem is formulated and model assumptions as well as notation are established. Section 3 states the main limit theorems and propositions for both the uncorrelated and correlated case. Due to their length, proofs of the theorems and propositions have been placed in the Appendix C. In Section 4 numerical examples based on the rough Bergomi model \cite{bg} are presented. Section 5 concludes.

\section{Assumptions and notation}

Consider a stochastic volatility model for the log-price of a stock
under a risk-neutral probability measure $P$:
\begin{equation}
\label{themodel}
X_{t}=X_0-\frac{1}{2}\int_{0}^{t}{\sigma _{s}^{2}}ds+\int_{0}^{t}\sigma
_{s}\left( \rho dW_{s}+\sqrt{1-\rho ^{2}}dB_{s}\right) ,\quad t\in \lbrack
0,T].  
\end{equation}%
Here, $X_0$ is the current log-price, $%
W$ and $B$ are standard Brownian motions defined on a complete probability
space  $(\Omega ,\mathcal{F},P)$, and $\sigma $ is a square-integrable and
right-continuous stochastic process adapted to the filtration generated by $%
W $.  We denote by $\mathcal{F}^{W}$ and $\mathcal{F}^{B}$
the filtrations generated by $W$ and $B$ and $\mathcal{F}:=%
\mathcal{F}^{W}\vee \mathcal{F}^{B}.$ We assume the interest rate $r$ to be zero for the sake of simplicity. The same arguments in this paper hold for $r\neq 0$.

There are two common flavours of forward starting options. Let $T$ denote the forward start date (also known as the fixing date), $\tau$ the expiry date, and $k$ the forward start log strike, then the payoff of the the first type of forward starting option is
\begin{equation}\label{type1}
(e^{X_{\tau}}- e^k e^{X_T})_+
\end{equation}
while the second type has payoff
\begin{equation} \label{type2}
(e^{X_{\tau}-X_T} - e^k )+
\end{equation}
Forward starting options with payoff given by \eqref{type1} are also known as Type II forward start options, and those with payoff given by \eqref{type2} are called Type I forward start options. In this note we will consider Type I options, and future reference to forward start options will mean options with payoff given by \eqref{type2} unless explicitly stated otherwise.

We assume $t\le T \le \tau$.
Under the above model, the price of a forward start European call at $t$ with strike price $K$, forward start date $T$ and expiry date $\tau$ is given by the equality
\[
V_{t}=E_{t}[(e^{X_\tau - X_{T}}-K)_{+}], 
\]%
where $E_{t}$ is the $\mathcal{F}_{t}-$conditional expectation with respect
to $P$ (i.e., $E_{t}[Z]=E[Z|\mathcal{F}_{t}]$). 
In the sequel, we make use
of the following notation:

\begin{itemize}
\item 
$v_{u}=\sqrt{\frac{Y_{u\vee T}}{\tau - (u\vee T)}}$, where $Y_s=\int_s^{\tau}\sigma_u^{2}du$.

That is, $v$
represents the future average volatility, and it is not an adapted process. 
Notice that $E_{t}\left[v_{T}\right] $ is the fair strike of a forward start volatility swap with maturity time $\tau$.

\item $BS(t,T,x,k,\sigma )$ is the price of a  European call option
under the Black-Scholes model with constant volatility $\sigma $,
stock price $e^x$, time to maturity $T-t,$ and strike $\exp(k)$. 
Remember that (if $r=0$)

\[
BS(t,T,x,k,\sigma )=e^{x}N(d_1(k,\sigma ))-e^{k}N(d_2(k,\sigma )), 
\]%
where $N$ denotes the cumulative probability function of the standard normal
law and
\[
d_1\left( k,\sigma \right) :=\frac{x-k}{\sigma \sqrt{T-t}}+
\frac{\sigma }{2}\sqrt{T-t},  \hspace{0.4cm} d_2\left( k,\sigma \right) :=\frac{x-k}{\sigma \sqrt{T-t}}-
\frac{\sigma }{2}\sqrt{T-t}.
\]%
For the sake of simplicty we make use of the notation $BS(k,\sigma):=BS(t,T,x,k,\sigma )$.
\item  The inverse function $BS^{-1}(t,T,x,k, \cdot)$ of the Black-Scholes formula with respect to the volatility parameter is defined as
\[
BS(t,T,x,k, BS^{-1}( t,T,x,k,\lambda) )=\lambda,
\]
for all $\lambda>0$.
For the sake of simplicity, we denote $BS^{-1}(k,\lambda)\ :=BS^{-1}(t,T,X_{t},k,\lambda)$.

Even though the Black-Scholes model's assumptions are violated, we can always express the stochastic volatility price of a forward starting option in terms of the forward start Black-Scholes price by 
\begin{equation}\label{fs=fsBMS}
V_{t}= 
BS \left(T,\tau,0,k,BS^{-1}(k,V_t) \right)
\end{equation}
As the forward start option does not depend on ${X_t}$, and $BS^{-1}(k,V_t)$ is also independent of the spot price.
\item 

For any fixed 
$t,T,\tau,k,$ we define the implied volatility $I(t,T,\tau,k)$ for the forward start European call option as the quantity such that
\[
BS( T,\tau,0,k,I( t,T,\tau,k) ) =V_{t}.
\]
Notice that $I(t,T,\tau,k)=BS^{-1}( T,\tau,0,k,V_t)$.

\item $\hat{k_t}$ is the {\it zero vanna implied volatility strike} of the forward starting option at time $t$. That is, the strike such that 
$$d_2(\hat{k}_t,I(t,T,\tau,\hat{k}_t))=0.$$ 
Moreover, we will refer to $I(t,T,\tau,\hat{k}_t)$ as the {\it zero vanna implied volatility} of the forward starting option.

\item 
$\Lambda _{r}:=E_{r}\left[ BS\left( T,\tau,0,k,v_{T}\right) \right]$.

\item 
$\Theta_{r}(k):=BS^{-1}(k,\Lambda_r)$. Notice that $\Theta_{t}(k)=I(t,T,\tau,k)$ and $\Theta_\tau(k)=v_T$.

\item $G(t,T,x,k,\sigma ):=( \frac{\partial ^{2}}{\partial x^{2}}-\frac{\partial}{\partial x}) BS(t,T,x,k,\sigma )$, 
and
$H(t,T,x,k,\sigma ):=( \frac{\partial ^{3}}{\partial x^{3}}-\frac{\partial ^{2}}{\partial x^{2}}) BS(t,T,x,k,\sigma )$.

\end{itemize}

In the remaining of this paper $\mathbb{D}_{W}^{1,2}$ denotes the domain of the
Malliavin derivative operator $D^{W}$ (see Appendix \ref{appendix1}) with respect to the Brownian motion $%
W. $ We also consider the
iterated derivatives $D^{n,W}$ , for $n>1,$ whose domains will be denoted by 
$\mathbb{D}_{W}^{n,2}$. We will use the notation $\mathbb{L}_{W}^{n,2}=$\ $%
L^{2}(\left[ 0,T\right] ;\mathbb{D}_{W}^{n,2})$. We refer to Nualart (2006) for a deeper introduction to this topic.
 
\section{Main results}

\subsection{The uncorrelated case}

Let us consider the following hypotheses:

\begin{description}
\item[(H1)] There exist two positive constants $a,b$ such that $a\leq \sigma
_{t}\leq b,$ for all $t\in \left[ 0,T\right] .$

\item[(H2)] $\sigma^2\in \mathbb{L}^{2,2}_W$ and there exist two constants $\nu>0$ and $H\in\left(0,1\right)$ such that, for all $0<r, \theta<s<T$,

$$
|D_r^W\sigma_s^2|\leq \nu(s-r)^{H-\frac12}, \hspace{0.3cm} |D_\theta^W D_r^W\sigma_s^2|\leq \nu^2(s-r)^{H-\frac12}(s-\theta)^{H-\frac12}.
$$

\end{description}

\begin{remark}
The above hypotheses have been chosen for the sake of simplicity, and could be replaced by adequate integrability conditions. For example, the constants could be substituted by random variales in $\cap_{p\geq 1}L^p(\Omega)$.
\end{remark}

The key tool in our analysis will be the following relationship between the zero vanna implied volatility and the fair strike of a volatility swap.

\begin{proposition}
\label{General}Consider the model (\ref{themodel}) with $\rho =0$ and assume
that hypotheses (H1) and (H2) hold. Then the zero vanna implied volatility of the forward starting option
admits the representation%
\begin{eqnarray}
\label{implied}
I\left( t,T,\tau,\hat{k}_t\right) &=& E_{t}\left[ v_{T}\right] 
+\frac12 E_t\bigg[\int_{t}^{\tau}\left( BS^{-1}\left( \hat{k}_t, \Lambda_r\right)\right) ^{\prime \prime\prime} \left(2 \int^\tau_r U_s D_r^W U_s ds\right) U_r dr\Bigg]\nonumber\\
&&+\frac12 E_t\Bigg[
\int_{t}^{\tau}\left( BS^{-1}\left( \hat{k}_t, \Lambda_r\right)\right) ^{(iv)} A_r U_r^2 dr \Bigg],
\label{I-E[v]}
\end{eqnarray}
where 
$$A_r:=\frac12 \int_{r}^{\tau} U_{s}^{2}ds, \hspace{0.3cm}  (D^-A)_r:= \frac{1}{2}\int^\tau_r D_r^W U^2_s ds,$$
and
\begin{eqnarray}
U_{r}&:=&E_{r}\left[D_r^W\left(BS(T,\tau,0,\hat{k}_t,v_T)\right)\right]\nonumber\\
&=&E_r\left[\frac{\partial BS}{\partial \sigma}(T,\tau,0,\hat{k}_t,v_T)\frac{1}{2v_T(\tau-T)}\int_{r\vee T}^\tau D_r^W\sigma_s^2 ds\right]
\end{eqnarray}
\end{proposition}

\vspace{0.2cm}

In order to prove our limit results, we will need the following hypothesis.

\begin{itemize}
\item (H2') $\sigma\in \mathbb{L}^{3,2}$ and there exists two constants $\nu>0$ and $H\in (0,1)$ such that, for all $0<r<u,s,\theta<\tau$
$$
|D_r^W\sigma_s^2|\leq \nu(\tau-r)^{H-\frac12}, \hspace{0.3cm} |D_\theta^W D_r^W\sigma_s^2|\leq \nu^2(\tau-r)^{H-\frac12}(\tau-\theta)^{H-\frac12},
$$
and 
$$
|D_u^W D_\theta^W D_r^W\sigma_s^2|\leq \nu^3(s-r)^{H-\frac12}(s-\theta)^{H-\frac12}(s-u)^{H-\frac12}.
$$
\end{itemize}
\begin{theorem}
\label{uncorrelated}
Consider the model (\ref{themodel}) and assume that hypotheses (H1) and (H2') hold. Then, 
\begin{eqnarray*}
\lefteqn{I\left( t,T,\tau,\hat{k}_t\right) -E_{t}\left[v_{T}\right]}\nonumber\\
&=&
O(\tau-T)^{4H+1} + O(\tau - T)^{2H+1}(\tau - t)^{2H} + O(\tau - T)(\tau - t)^{4H}.
\end{eqnarray*}

\end{theorem}


\subsection{The correlated case}

We will consider the following hypothesis.

\begin{description}
\item[(H3)] Hypotheses (H1), (H2'), hold and 
terms 
\[
\frac{1}{(\tau-T)^{3+2H}}E_{t} \left[\left(\int_{T}^{\tau}
\int_{s}^{\tau}D_{s}^{W}\sigma _{r}^{2}dr ds\right) ^{2}\right] ,
\]%
\[
\frac{1}{(\tau-T)^{2+2H }}E_{t}\left[\int_{T}^{\tau}\left( \int_{s}^{\tau} D_{s}^{W}\sigma
_{r} dr\right) ^{2}ds\right],
\]%
\[
\frac{1}{(\tau-T)^{2+2H }}E_t \left[ \int_{T}^{\tau} \left( E_r\left[\int_r^\tau D_r^W\sigma_s^2 ds\right]\right)^2 dr \right],
\]
and%
\[
\frac{1}{(\tau-T)^{2+2H }}E_{t}\left[
\int_{T}^{\tau} \int_{s}^{\tau} \int_{r}^{\tau}D_{s}^{W}D_{r}^{W}\sigma_{u}^{2}du dr ds\right],
\]%
have a finite limit as $\tau\rightarrow T.$

\end{description}

The following result is an adaptation of Proposition 4.1  in Al\`os and Shiraya (2019) (see Appendix \ref{appendix12}) and it gives us an exact decomposition for the zero vanna implied volatility of the forward starting option that will be the main tool in this Section.
\begin{proposition}
\label{Theoremcorrelatedcase2}
Consider the model \eqref{themodel} and assume
that hypotheses (H1), (H2') and hold for some $H\in (0,1)$. Then, for every $k\in \mathbb{R}$

\begin{eqnarray}
I\left( t,T,\tau,\hat{k}_t\right) &=&I^{0}\left( t,T,\tau,k\right)
\nonumber \\
&&+\frac{\rho }{2}E_t \left[\int_{T}^{\tau}( BS^{-1}) ^{\prime }( \hat{k}_t,\Gamma _{s}) H(s,\tau,X_{s}-X_T,k,v_{s})\zeta_{s}ds\right],
\label{impliedrel}
\end{eqnarray}

where $I^{0}( t,T,\tau,k) $ denotes the zero vanna implied
volatility of the forward starting option in the uncorrelated case $\rho =0$,

\begin{eqnarray*}
\Gamma _{s}&:=&
E_t[BS(T,\tau,0,k,v_T)]\nonumber\\
&&+\frac{\rho}{2}E_t\left[\int_T^s H(u,\tau,,X_u-X_T,k,v_u)\zeta_u du\right].
\end{eqnarray*}

and 

$\zeta _{t}:=\sigma _{t}\int_{t}^{\tau}D_{t}^{W}\sigma _{r}^{2}dr.$ 

\end{proposition}

\begin{proof}
The process $E_u(e^{X_\tau-X_T})$ is a stochastic volatility model with volatility equal to
$\sigma_u\bf{1}_{[T,\tau]}(u)$ such that it is equatl to $e^{X_u-X_T}$ for $u>T$ and to $1$ if $u<T$. Then, the result follows as a direct application of Proposition 4.1 in Al\`os and Shiraya (Proposition \ref{Theoremcorrelatedcase} in Appendix B) and the fact that $BS(t,\tau,0,k,v_t)=BS(T,\tau,0,k,v_T)$.

\end{proof}

Theorem \ref{uncorrelated} and Proposition \ref{Theoremcorrelatedcase2} allow us to prove the following result.

\begin{theorem}
\label{themaintheorem}
Consider the model \eqref{themodel}\ and assume
that hypotheses (H1), (H2') and (H3)  hold for some $H\in(0,\frac12) $. Then%
\begin{eqnarray}
\label{lim1}
\lefteqn{\lim_{\tau\rightarrow T}\frac{I( t,T,\tau,\hat{k}_t) -E_{t}[v_{T}]} {(\tau-T)^{2H }}}\nonumber \\
&=& \lim_{\tau\rightarrow T}\frac{3\rho^{2}}{8(\tau-T)^{3+2H}}
E_{t}\left[\frac{1}{\sigma _{T}^{3}}\left(\int_{T}^{\tau}\int_{s}^\tau D_{s}^{W}\sigma _{r}^{2}dr ds\right)^{2} \right]\nonumber\\
&&-\lim_{\tau\rightarrow T}\frac{\rho ^{2}}{2(\tau-T)^{2+2H}}
E_{t}\left[\frac{1}{\sigma _{T}}\int_{T}^{\tau}\left(\int_{s}^\tau D_{s}^{W}\sigma_{r} dr\right) ^{2}ds\right]\nonumber\\
&&-\lim_{\tau\rightarrow T}\frac{\rho ^{2}}{2(\tau-T)^{2+2H}}
E_{t}\left[\frac{1}{\sigma _{T}}\int_{T}^{\tau} \int_{s}^\tau\int_{r}^\tau D_{s}^{W}D_{r}^{W}\sigma _{u}^{2}du dr ds\right].
\end{eqnarray}
\end{theorem}

\begin{corollary}
Assume that $\sigma_t=f(B^H_t)$, where $f\in\mathcal{C}_b^3$ with a range in a compact set of $\mathbb{R}^+$  and $B^H_t$ is a fBm with Hurst parameter $H$. Then 
the above result proves that, in the correlated case
$$I(t,T,\tau, X_t,\hat{k}_t) -E_{t}[ v_{t}]=O((\tau-T)^{2H}).$$
\end{corollary}

\begin{remark}
Notice that the term $T_2^{1,2}$ in Appendix \ref{appendix2} is of the order $(\rho(T-t)^{\frac12+{2}H})$. When $T-t$ does not tend to zero, this term can not be neglected.
\end{remark}

\begin{remark}
Hypotheses (H1)-(H3) have been chosen for the sake of simplicity. The same results can be extended to other stochastic volatility models (see e.g., Section 5 in Al\`{o}s and Shiraya (2019)).
\end{remark}


\section{Numerical examples}

We check our estimates for the forward starting volatility swap numerically. To examine the accuracy of the approximation, we assume the following rough Bergomi model as in Bayer, Friz and Gatheral (2016):
\begin{eqnarray}
S_t &=& \exp\left( X_0- \frac{1}{2} \int_{0}^{t} \sigma _{s}^{2} ds+\int_{0}^{t}\sigma
_{s}dB_{s}\right),\\
\sigma_t^2 &=&\sigma_0^2\exp\left(\alpha W_t^H-\frac12 \alpha^2t^{2H}\right),\ t\in[0,T],
\end{eqnarray}
where $\sigma_0$ and $\alpha$ are positive constants, $H\in (0,1)$, $W_t^H:=\sqrt{2H}\int_0^t (t-s)^{H-\frac12}dW_s$, and $W$ and $B$ are standard Brownian motions.

For all $s<t$
\begin{eqnarray}
\label{rBcov}
E(W_t^H W_s^H)
&=& 
s^{2H} \int^1_0 \frac{2H}{(1-x)^{\frac12 - H}(t/s-x)^{\frac12 - H}}dx \nonumber \\
&=&
s^{2H}2H\left(\frac{1}{(1-(1/2 -H))(t/s)^{(1/2 -H)}}+ \frac{1-1/2H}{1-(1-1/2H)}\int^1_0 \frac{(1-x)^{1-(1-1/2H)}}{(t/s - x)^{(1-1/2H)+1}}dx\right),
\end{eqnarray}
and for all $s, t \ge 0$,
\begin{equation}
\label{rBW}
E(W_t^H B_s)=\frac{\rho\sqrt{2H}}{H+\frac12}\left(t^{H+\frac12}-(t-\min(t,s))^{H+\frac12}\right).
\end{equation}
where $\rho \in [-1,1]$.

We set the parameters $\sigma _{0} = 20\%$, $\alpha =0.8$, $\rho =0$ or $-0.8$, and the Hurst parameters $H=0.05$, $0.1$, $0.3$.
The case $\alpha=2$ is also numerically simulated in order to test the accuracy of the approximation for parameter values commonly encountered in equity markets.

To calculate the implied volatility for the forward start European option and forward volatility swap strike, we use Monte Carlo simulation, whose number of time steps is 250 for one year and the number of simulations is ten million.
We use the Black-Scholes model as the control variate for the Monte Carlo simulation.
After obtaining the exact forward start volatility swap strikes and options prices, we calculate the ATM and zero vanna implied volatilities by the bisection method.

Tables \ref{t1} to \ref{t4} show the results. In the tables, $E_t[v_T]$ is the simulated exact forward start volatility swap value,
and $I(\hat{k})$ and ATMI are the implied volatilities at the zero vanna strike and ATM strike for the forward starting European option, respectively.
$T-t$ is the interval from contract's time to the forward stating time, and $\tau - T$ shows the interval from the forward starting time to the maturity.

\begin{table}[H]
\centering
\begin{adjustbox}{max width=\textwidth}
\begin{tabular}{ccrrrrrrrrr} \hline
 & $T-t$ & \multicolumn{1}{c}{0.5} & \multicolumn{1}{c}{0.5} & \multicolumn{1}{c}{0.5} & \multicolumn{1}{c}{1} & \multicolumn{1}{c}{1} & \multicolumn{1}{c}{1} & \multicolumn{1}{c}{2} & \multicolumn{1}{c}{2} & \multicolumn{1}{c}{2} \\
$H$ index & $\tau-T$ & \multicolumn{1}{c}{0.5} & \multicolumn{1}{c}{1} & \multicolumn{1}{c}{2} & \multicolumn{1}{c}{0.5} & \multicolumn{1}{c}{1} & \multicolumn{1}{c}{2} & \multicolumn{1}{c}{0.5} & \multicolumn{1}{c}{1} & \multicolumn{1}{c}{2} \\\hline
$0.05$ & $E_t[v_T]$ & 19.58\% & 19.62\% & 19.65\% & 19.50\% & 19.55\% & 19.59\% & 19.40\% & 19.46\% & 19.52\% \\
 & $I (\hat{k})$ & 19.57\% & 19.62\% & 19.64\% & 19.49\% & 19.55\% & 19.59\% & 19.40\% & 19.46\% & 19.51\% \\
 & ATMI & 19.57\% & 19.62\% & 19.64\% & 19.49\% & 19.55\% & 19.58\% & 19.40\% & 19.45\% & 19.50\% \\\cline{2-11}
 & $I (\hat{k})$ $-$ $E_t[v_T]$ & 0.00\% & 0.00\% & 0.00\% & 0.00\% & 0.00\% & -0.01\% & 0.00\% & 0.00\% & -0.01\% \\
 & ATMI $-$ $E_t[v_T]$ & 0.00\% & 0.00\% & -0.01\% & 0.00\% & 0.00\% & -0.01\% & 0.00\% & -0.01\% & -0.02\% \\\hline
$0.1$ & $E_t[v_T]$ & 19.32\% & 19.36\% & 19.38\% & 19.16\% & 19.23\% & 19.27\% & 18.97\% & 19.04\% & 19.11\% \\
 & $I (\hat{k})$ & 19.32\% & 19.36\% & 19.37\% & 19.16\% & 19.23\% & 19.26\% & 18.97\% & 19.04\% & 19.10\% \\
 & ATMI & 19.32\% & 19.36\% & 19.36\% & 19.16\% & 19.22\% & 19.25\% & 18.96\% & 19.03\% & 19.09\% \\\cline{2-11}
 & $I (\hat{k})$$-$ $E_t[v_T]$ & 0.00\% & 0.00\% & 0.00\% & 0.00\% & 0.00\% & -0.01\% & 0.00\% & -0.01\% & -0.01\% \\
 & ATMI $-$ $E_t[v_T]$ & -0.01\% & 0.00\% & -0.01\% & -0.01\% & -0.01\% & -0.02\% & 0.00\% & -0.01\% & -0.03\% \\\hline
$0.3$ & $E_t[v_T]$ & 18.97\% & 18.88\% & 18.69\% & 18.52\% & 18.46\% & 18.32\% & 17.81\% & 17.79\% & 17.71\% \\
 & $I (\hat{k})$ & 18.97\% & 18.88\% & 18.69\% & 18.51\% & 18.46\% & 18.31\% & 17.81\% & 17.78\% & 17.69\% \\
 & ATMI & 18.97\% & 18.87\% & 18.67\% & 18.51\% & 18.44\% & 18.28\% & 17.80\% & 17.76\% & 17.65\% \\\cline{2-11}
 &$I (\hat{k})$$-$ $E_t[v_T]$ & 0.00\% & 0.00\% & 0.00\% & 0.00\% & -0.01\% & -0.01\% & 0.00\% & -0.01\% & -0.02\% \\
 & ATMI $-$ $E_t[v_T]$ & -0.01\% & -0.01\% & -0.03\% & -0.01\% & -0.02\% & -0.04\% & -0.01\% & -0.03\% & -0.06\% \\\hline
\end{tabular}
\end{adjustbox}
\caption{ ($\rho=0$, $\alpha = 0.8$)}
\label{t1}
\end{table}

\begin{table}[H]
\centering
\begin{adjustbox}{max width=\textwidth}
\begin{tabular}{ccrrrrrrrrr} \hline
 & $T-t$ & \multicolumn{1}{c}{0.5} & \multicolumn{1}{c}{0.5} & \multicolumn{1}{c}{0.5} & \multicolumn{1}{c}{1} & \multicolumn{1}{c}{1} & \multicolumn{1}{c}{1} & \multicolumn{1}{c}{2} & \multicolumn{1}{c}{2} & \multicolumn{1}{c}{2} \\
$H$ index & $\tau-T$ & \multicolumn{1}{c}{0.5} & \multicolumn{1}{c}{1} & \multicolumn{1}{c}{2} & \multicolumn{1}{c}{0.5} & \multicolumn{1}{c}{1} & \multicolumn{1}{c}{2} & \multicolumn{1}{c}{0.5} & \multicolumn{1}{c}{1} & \multicolumn{1}{c}{2} \\\hline
$0.05$ & $E_t[v_T]$ & 19.58\% & 19.62\% & 19.65\% & 19.50\% & 19.55\% & 19.59\% & 19.40\% & 19.46\% & 19.52\% \\
 &$I (\hat{k})$ & 19.45\% & 19.48\% & 19.50\% & 19.37\% & 19.42\% & 19.45\% & 19.27\% & 19.33\% & 19.37\% \\
 & ATMI & 19.30\% & 19.27\% & 19.19\% & 19.23\% & 19.21\% & 19.14\% & 19.13\% & 19.12\% & 19.06\% \\\cline{2-11}
 &$I (\hat{k})$$-$ $E_t[v_T]$ & -0.13\% & -0.13\% & -0.14\% & -0.12\% & -0.13\% & -0.15\% & -0.12\% & -0.13\% & -0.15\% \\
 & ATMI $-$ $E_t[v_T]$ & -0.27\% & -0.35\% & -0.46\% & -0.27\% & -0.34\% & -0.46\% & -0.27\% & -0.35\% & -0.46\% \\\hline
$0.1$ & $E_t[v_T]$ & 19.32\% & 19.36\% & 19.38\% & 19.16\% & 19.23\% & 19.27\% & 18.97\% & 19.04\% & 19.11\% \\
 &$I (\hat{k})$& 19.13\% & 19.14\% & 19.13\% & 18.98\% & 19.02\% & 19.02\% & 18.78\% & 18.82\% & 18.85\% \\
 & ATMI & 18.96\% & 18.88\% & 18.73\% & 18.81\% & 18.76\% & 18.62\% & 18.61\% & 18.57\% & 18.46\% \\\cline{2-11}
 &$I (\hat{k})$ $-$ $E_t[v_T]$ & -0.19\% & -0.22\% & -0.25\% & -0.18\% & -0.21\% & -0.25\% & -0.19\% & -0.22\% & -0.26\% \\
 & ATMI $-$ $E_t[v_T]$ & -0.36\% & -0.48\% & -0.64\% & -0.36\% & -0.47\% & -0.64\% & -0.36\% & -0.48\% & -0.65\% \\\hline
$0.3$ & $E_t[v_T]$& 18.97\% & 18.88\% & 18.69\% & 18.52\% & 18.46\% & 18.32\% & 17.81\% & 17.79\% & 17.71\% \\
 & $I (\hat{k})$& 18.77\% & 18.58\% & 18.24\% & 18.32\% & 18.16\% & 17.86\% & 17.61\% & 17.47\% & 17.23\% \\
 & ATMI & 18.60\% & 18.29\% & 17.77\% & 18.15\% & 17.88\% & 17.40\% & 17.45\% & 17.21\% & 16.79\% \\\cline{2-11}
 & $I (\hat{k})$ $-$ $E_t[v_T]$& -0.20\% & -0.30\% & -0.45\% & -0.20\% & -0.30\% & -0.46\% & -0.21\% & -0.32\% & -0.48\% \\
 & ATMI $-$ $E_t[v_T]$& -0.37\% & -0.59\% & -0.93\% & -0.36\% & -0.58\% & -0.92\% & -0.36\% & -0.58\% & -0.92\% \\\hline
\end{tabular}
\end{adjustbox}
\caption{Zero vanna implied volatility approximation ($\rho=-0.8$, $\alpha = 0.8$)}
\label{t2}
\end{table}

\begin{table}[H]
\centering
\begin{adjustbox}{max width=\textwidth}
\begin{tabular}{ccrrrrrrrrr} \hline
 & $T-t$ & \multicolumn{1}{c}{0.5} & \multicolumn{1}{c}{0.5} & \multicolumn{1}{c}{0.5} & \multicolumn{1}{c}{1} & \multicolumn{1}{c}{1} & \multicolumn{1}{c}{1} & \multicolumn{1}{c}{2} & \multicolumn{1}{c}{2} & \multicolumn{1}{c}{2} \\
$H$ index & $\tau-T$ & \multicolumn{1}{c}{0.5} & \multicolumn{1}{c}{1} & \multicolumn{1}{c}{2} & \multicolumn{1}{c}{0.5} & \multicolumn{1}{c}{1} & \multicolumn{1}{c}{2} & \multicolumn{1}{c}{0.5} & \multicolumn{1}{c}{1} & \multicolumn{1}{c}{2} \\\hline
$0.05$ & $E_t[v_T]$ & 17.31\% & 17.57\% & 17.74\% & 16.86\% & 17.20\% & 17.44\% & 16.35\% & 16.72\% & 17.04\% \\
 &$I (\hat{k})$& 17.30\% & 17.57\% & 17.73\% & 16.85\% & 17.19\% & 17.42\% & 16.34\% & 16.70\% & 17.01\% \\
 & ATMI & 17.29\% & 17.55\% & 17.70\% & 16.84\% & 17.17\% & 17.38\% & 16.33\% & 16.67\% & 16.97\% \\\cline{2-11}
 & $I (\hat{k})$ $-$ $E_t[v_T]$ & -0.01\% & 0.00\% & -0.01\% & -0.01\% & -0.01\% & -0.03\% & 0.00\% & -0.02\% & -0.03\% \\
 & ATMI $-$ $E_t[v_T]$ & -0.02\% & -0.02\% & -0.05\% & -0.02\% & -0.03\% & -0.07\% & -0.02\% & -0.04\% & -0.08\% \\\hline
$0.1$ &$E_t[v_T]$ & 15.99\% & 16.17\% & 16.23\% & 15.18\% & 15.48\% & 15.66\% & 14.22\% & 14.58\% & 14.88\% \\
 & $I (\hat{k})$& 15.98\% & 16.16\% & 16.20\% & 15.16\% & 15.45\% & 15.60\% & 14.21\% & 14.53\% & 14.80\% \\
 & ATMI & 15.96\% & 16.13\% & 16.14\% & 15.15\% & 15.42\% & 15.54\% & 14.19\% & 14.50\% & 14.74\% \\\cline{2-11}
 & $I (\hat{k})$ $-$ $E_t[v_T]$ & -0.01\% & -0.01\% & -0.03\% & -0.02\% & -0.02\% & -0.06\% & -0.02\% & -0.04\% & -0.07\% \\
 & ATMI $-$ $E_t[v_T]$ & -0.03\% & -0.04\% & -0.09\% & -0.03\% & -0.05\% & -0.11\% & -0.03\% & -0.08\% & -0.13\% \\\hline
$0.3$ & $E_t[v_T]$ & 14.35\% & 13.90\% & 13.09\% & 12.32\% & 12.07\% & 11.50\% & 9.67\% & 9.57\% & 9.26\% \\
 & $I (\hat{k})$ & 14.33\% & 13.86\% & 12.97\% & 12.27\% & 11.99\% & 11.32\% & 9.60\% & 9.42\% & 9.02\% \\
 & ATMI & 14.31\% & 13.83\% & 12.91\% & 12.26\% & 11.95\% & 11.26\% & 9.58\% & 9.39\% & 8.97\% \\\cline{2-11}
 & $I (\hat{k})$ $-$$E_t[v_T]$& -0.02\% & -0.04\% & -0.13\% & -0.04\% & -0.09\% & -0.18\% & -0.07\% & -0.14\% & -0.24\% \\
 & ATMI $-$ $E_t[v_T]$ & -0.04\% & -0.08\% & -0.18\% & -0.06\% & -0.12\% & -0.24\% & -0.09\% & -0.17\% & -0.29\% \\\hline
\end{tabular}
\end{adjustbox}
\caption{ ($\rho=0$, $\alpha=2$)}
\label{t3}
\end{table}

\begin{table}[H]
\centering
\begin{adjustbox}{max width=\textwidth}
\begin{tabular}{ccrrrrrrrrr} \hline
 & $T-t$ & \multicolumn{1}{c}{0.5} & \multicolumn{1}{c}{0.5} & \multicolumn{1}{c}{0.5} & \multicolumn{1}{c}{1} & \multicolumn{1}{c}{1} & \multicolumn{1}{c}{1} & \multicolumn{1}{c}{2} & \multicolumn{1}{c}{2} & \multicolumn{1}{c}{2} \\
$H$ index & $\tau-T$ & \multicolumn{1}{c}{0.5} & \multicolumn{1}{c}{1} & \multicolumn{1}{c}{2} & \multicolumn{1}{c}{0.5} & \multicolumn{1}{c}{1} & \multicolumn{1}{c}{2} & \multicolumn{1}{c}{0.5} & \multicolumn{1}{c}{1} & \multicolumn{1}{c}{2} \\\hline
$0.05$ & $E_t[v_T]$  & 17.31\% & 17.57\% & 17.74\% & 16.86\% & 17.20\% & 17.44\% & 16.35\% & 16.72\% & 17.04\% \\
 & $I (\hat{k})$& 16.85\% & 17.07\% & 17.20\% & 16.41\% & 16.70\% & 16.89\% & 15.89\% & 16.22\% & 16.47\% \\
 & ATMI & 16.63\% & 16.74\% & 16.72\% & 16.19\% & 16.38\% & 16.42\% & 15.69\% & 15.92\% & 16.02\% \\\cline{2-11}
 & $I (\hat{k})$ $-$ $E_t[v_T]$  & -0.46\% & -0.50\% & -0.55\% & -0.45\% & -0.50\% & -0.55\% & -0.46\% & -0.50\% & -0.57\% \\
 & ATMI $-$ $E_t[v_T]$  & -0.68\% & -0.83\% & -1.02\% & -0.66\% & -0.81\% & -1.02\% & -0.66\% & -0.80\% & -1.02\% \\\hline
$0.1$ & $E_t[v_T]$  & 15.99\% & 16.17\% & 16.23\% & 15.18\% & 15.48\% & 15.66\% & 14.22\% & 14.58\% & 14.88\% \\
 & $I (\hat{k})$ & 15.23\% & 15.29\% & 15.22\% & 14.44\% & 14.61\% & 14.65\% & 13.48\% & 13.71\% & 13.85\% \\
 & ATMI & 14.99\% & 14.93\% & 14.69\% & 14.21\% & 14.27\% & 14.15\% & 13.29\% & 13.41\% & 13.40\% \\\cline{2-11}
 & $I (\hat{k})$ $-$ $E_t[v_T]$  & -0.76\% & -0.88\% & -1.01\% & -0.75\% & -0.87\% & -1.01\% & -0.74\% & -0.86\% & -1.03\% \\
 & ATMI $-$ $E_t[v_T]$  & -1.00\% & -1.24\% & -1.53\% & -0.97\% & -1.20\% & -1.51\% & -0.93\% & -1.17\% & -1.48\% \\\hline
$0.3$ & $E_t[v_T]$  & 14.35\% & 13.90\% & 13.09\% & 12.32\% & 12.07\% & 11.50\% & 9.67\% & 9.57\% & 9.26\% \\
 & $I (\hat{k})$ & 13.45\% & 12.65\% & 11.45\% & 11.47\% & 10.89\% & 9.94\% & 8.91\% & 8.50\% & 7.86\% \\
 & ATMI & 13.24\% & 12.34\% & 11.04\% & 11.31\% & 10.66\% & 9.62\% & 8.81\% & 8.35\% & 7.65\% \\\cline{2-11}
 & $I (\hat{k})$ $-$ $E_t[v_T]$  & -0.90\% & -1.25\% & -1.65\% & -0.84\% & -1.18\% & -1.56\% & -0.76\% & -1.06\% & -1.40\% \\
 & ATMI $-$ $E_t[v_T]$  & -1.11\% & -1.56\% & -2.05\% & -1.00\% & -1.42\% & -1.88\% & -0.86\% & -1.21\% & -1.61\% \\\hline
\end{tabular}
\end{adjustbox}
\caption{Zero vanna implied volatility approximation ($\rho=-0.8$, $\alpha=2$)}
\label{t4}
\end{table}

In all cases, the zero vanna implied volatility for the forward starting European option approximates the forward start volatility swap strike better than ATM implied volatility for the forward starting European option.
Also, as expected, the zero vanna approximation for the forward start volatility swap is accurate in the uncorrelated case compared with the correlated case.
It is because the error order in the uncorrelated case is higher than that of the correlated case.
Moreover, since the error terms depend on the terms of the forward starting time and the maturity,
the small $\tau-T$ results are better than those of large $\tau-T$.
Comparing the effects of the size of $T-t$ and the size of $\tau-T$ on accuracy, our results show $\tau-T$ affects greater than $T-t$. 
Regarding the Hurst parameter, as the parameter increases, the order on the time increases,
and the approximation errors in short terms becomes smaller as shown in Theorems \ref{uncorrelated} and \ref{themaintheorem}.


\section{Conclusion}
A rigorous quantification of the approximation error of the forward start zero vanna implied volatility has been given in rough volatility models. The results are an extension to the forward start case of previous results on the spot starting case. For relatively tractable rough volatility models, such as the rBergomi model, our results can be used to accurately price short term forward volatility swaps without having to resort to numerical schemes. Furthermore, the result can also be employed to provide lower bounds for VIX futures prices.

\appendix
\section{Malliavin calculus}\label{appendix1}
In this appendix, we present the basic Malliavin calculus results we use in this paper. The first one is the Clark-Ocone-Haussman formula, that allows us to compute explicitly the martingale representation of a random variable $F\in \mathbb{D}^{1,2}_W$.

\begin{theorem}[ Clark-Ocone-Harussman formula]Consider a Brownian motion $W=\{W_t, t\in [0,T]\}$ and a random variable $F\in \mathbb{D}^{1,2}_W$. Then
$$
F=E[F]+\int_0^T E_r[D_r^WF] dW_r.
$$
\end{theorem}

We will also make use of the following anticipating It\^o's formula (see for example, Al\`os and Garc\'ia-Lorite (2021)), that allows us to work with non-adapted processes.

\begin{proposition}
\label{Ito}
Assume model \eqref{themodel} and $\sigma ^{2}\in \mathbb{L}^{1,2}_W$. Let
$F:[0,T]\times \mathbb{R}^{2}\rightarrow \mathbb{R}$ be a function
in $C^{1,2} ([0,T]\times \mathbb{R}^{2})$ such that there exists a
positive constant $C$ such that, for all $t\in [ 0,T]
,$ $F$ and its partial derivatives evaluated in $(
t,X_{t},Y_{t})$ are bounded by $C.$ Then it follows that
\begin{eqnarray}
F(t,X_{t},Y_{t}) &=&F(0,X_{0},Y_{0})+\int_{0}^{t}{\partial _{s}F}%
(s,X_{s},Y_{s})ds \nonumber\\
&&-\int_{0}^{t}{\partial _{x}F}(s,X_{s},Y_{s}) \frac{\sigma_{s}^{2}}{2}ds \nonumber\\
&&+\int_{0}^{t}{\partial _{x}F}(s,X_{s},Y_{s})\sigma _{s}(\rho dW_{s}+\sqrt{%
1-\rho ^{2}}dB_{s}) \nonumber\\
&&-\int_{0}^{t}{\partial _{y}F}(s,X_{s},Y_{s})\sigma
_{s}^{2}ds+\rho
\int_{0}^{t}{\partial _{xy}^{2}F}(s,X_{s},Y_{s})\Xi _{s}ds \nonumber\\
&&+\frac{1}{2}\int_{0}^{t}{\partial
_{xx}^{2}F}(s,X_{s},Y_{s})\sigma_{s}^{2}ds ,
\label{aito}
\end{eqnarray}
where $\Xi _{s}:=(\int_{s}^{t}D^W_{s}\sigma _{r}^{2}dr)\sigma _{s}.$ 
\end{proposition}

\section{Previous results}\label{appendix12}
Here we recall this decomposition formula for the implied volatility in the correlated case.
\begin{proposition}[Theorem 9 in Al\`os and Shiraya (2019)]
\label{Theoremcorrelatedcase}Consider the model \eqref{themodel} and assume
that hypotheses (H1) and (H2') hold for some $H \in
(0,1) $. Then, for every $k\in \mathbb{R}$,
\begin{eqnarray}
I\left( t,T,X_t,k\right) &=&I^{0}\left( t,T,X_t,k\right)
\nonumber \\
&&\hspace{-0.5cm}+\frac{\rho }{2}E_t \left[\int_{t}^{T}( BS^{-1}) ^{\prime }( k,\Gamma _{s}) H(s,T,X_{s},k,v_{s})\Phi _{s}ds\right],
\label{impliedrel}
\end{eqnarray}%
where $I^{0}( t,T,X_t,k) $ denotes the implied
volatility in the uncorrelated case $\rho =0$,
\[
\Gamma _{s}:=E_{t}[ BS(t,T,X_t,k,v_{t})] +\frac{\rho }{2}%
E_{t}\left[\int_{t}^{s} H(r,T,X_{r},k,v_{r})\Phi _{r}dr\right], 
\]%
and $\Phi _{t}:=\sigma _{t}\int_{t}^{T}D_{t}^{W}\sigma _{r}^{2}dr.$ 
\end{proposition}

\section{Proofs}\label{appendix2}

\begin{proof}[Proof of Proposition \ref{General}]
This proof is decomposed in several steps.\\

\noindent {\it Step 1} Firstly, we show that
\begin{equation}
\label{primerpas}
I\left( t,T,\tau,\hat{k}_t\right)  =E_{t}\left[ v_{T}\right] +\frac{1}{2} E_t
\left[\int_{t}^{\tau} \left( BS^{-1}\right) ^{\prime \prime}\left( \hat{k}_t, \Lambda_r\right)U_{r}^{2}dr\right].
\end{equation}
In $\rho=0$, the Hull and White formula gives $V_t=\Lambda_t$. Then, as in the proof of Proposition 3.1 in Al\`os and Shiraya (2019),
\begin{equation}
\label{impliedmartingale}
I(t,T,\tau,\hat{k}_t)=BS^{-1}(\hat{k}_t,\Lambda_t)=E_t[BS^{-1}(\hat{k}_t,\Lambda_t)].
\end{equation}
From (H2) and the Clark-Ocone formula (see Appendix \ref{appendix1}),  $\Lambda$ has the martingale representation given by
\begin{eqnarray}
d\Lambda_r&=&E_r[D_r^W(BS(t,T,\tau,0,\hat{k}_t,v_T)]dW_r \nonumber\\
&=&E_r\left[\frac{\partial BS}{\partial \sigma}(t,T,\tau,0,\hat{k}_t,v_T)\frac{1}{2v_T(\tau-T)}\int_{r\vee T}^\tau D_r^W\sigma_s^2 ds\right]dW_r \nonumber\\
&=&U_rdW_r.
\end{eqnarray}
After taking expectations, the classical It\^o's formula gives
\begin{eqnarray}
\label{Itoimplied}
E_t[BS^{-1}(\hat{k}_t,\Lambda_t)]
=E_t[BS^{-1}(\hat{k}_t,\Lambda_\tau)]-\frac12 E_t \left[ \int_{t}^{\tau} \left( BS^{-1}\right) ^{\prime \prime}\left( \hat{k}_t, \Lambda_r\right)d\langle \Lambda,\Lambda\rangle_r\right].
\end{eqnarray}
Now, 
$\Lambda_\tau=BS\left( t,T,\tau,0,\hat{k}_t,v_{T}\right)$, (\ref{impliedmartingale}) and (\ref{Itoimplied}) imply that
\begin{eqnarray}
\label{Itoimplied}
I(t,T,\tau,\hat{k}_t)=E_{t}\left[ v_{T}\right] -\frac12 E_t \left[ \int_{t}^{\tau} \left( BS^{-1}\right) ^{\prime \prime}\left( \hat{k}_t, \Lambda_r\right)d\langle \Lambda,\Lambda\rangle_r\right].
\end{eqnarray}
That is, $$
I(t,T,\tau,\hat{k}_t)=E_{t}\left[ v_{T}\right] -\frac12 E_t \left[ \int_{t}^{\tau} \left( BS^{-1}\right) ^{\prime \prime}\left( \hat{k}_t, \Lambda_r\right)U_r^2dr\right].
$$
{\it Step 2} Let us see
\begin{eqnarray}
\label{secondstep}
\lefteqn{E_t \left[\int_{t}^{\tau} \left( BS^{-1}\right) ^{\prime \prime}\left(\hat{k}_t, \Lambda_r\right)U_{r}^{2}dr\right]}\nonumber\\
&=&E_t\Bigg[\int_{t}^{\tau}\left( BS^{-1}\left( \hat{k}_t, \Lambda_r\right)\right) ^{\prime \prime\prime} (D^-A)_r U_r dr\Bigg]\nonumber\\
&&+\frac12E_t\Bigg[
\int_{t}^{\tau}\left( BS^{-1}\left( \hat{k}_t, \Lambda_r\right)\right) ^{(iv)} A_r U_r^2 dr \Bigg].
\label{I-E[v]}
\end{eqnarray}
To this end, we apply the anticipating It\^o's formula (see Appendix \ref{appendix1}) to 
$$
 \left( BS^{-1}\right) ^{\prime \prime}\left( \hat{k}_t, \Lambda_r\right)A_r,
$$
and, taking this into account that $dA_r=-\frac12U_r^2dr$, we get 
\begin{eqnarray}
E\bigg[
\left( BS^{-1}\right) ^{\prime \prime}\left( \hat{k}_t, \Lambda_\tau \right)A_\tau
\bigg]
&=&E_t \bigg[\left( BS^{-1}\right) ^{\prime \prime}\left( \hat{k}_t, \Lambda_t\right)A_t \bigg]\nonumber\\
&&-\frac{1}{2} E_t\Bigg[
\int_{t}^{\tau} \left( BS^{-1}\right) ^{\prime \prime}\left( \hat{k}_t, \Lambda_r\right)U_{r}^{2}dr\Bigg]\nonumber\\
&&+\frac12 E_t\bigg[\int_{t}^{\tau}\left( BS^{-1}\left( \hat{k}_t, \Lambda_r\right)\right) ^{\prime \prime\prime} \left(\int^\tau_r D_r^W U_s^2 ds\right) U_r dr\Bigg]\nonumber\\
&&+\frac12 E_t\Bigg[
\int_{t}^{\tau}\left( BS^{-1}\left( \hat{k}_t, \Lambda_r\right)\right) ^{(iv)} A_r U_r^2 dr \Bigg]
\end{eqnarray}
In particular, 
$$\left( BS^{-1}\right) ^{\prime \prime }\left(\hat{k}_t,{\Lambda _{t}}\right)=0,$$
and 
\begin{eqnarray}
\left( BS^{-1}\right) ^{\prime \prime }\left(k_t,{\Lambda _{r}}%
\right) &=&
\frac{(\Theta_r(k))^4(\tau-T)^2-4(1-k_t)^2}{4\left( \exp(X_{t})N^{\prime }(d_{+}\left( k_t,\Gamma^t_r\right))(\tau-T) \right)^{2} (\Theta_r(k))^3}
\end{eqnarray}
and
\begin{eqnarray}
\left( BS^{-1}\right) ^{\prime \prime }\left(\hat{k}_t,{\Lambda _{\tau}}\right)
=\frac{(v_T^4-(I(t,T,\tau,\hat{k}_t))^4)}{4\left( \exp(1)N^{\prime }(d_{+}\left( \hat{k}_t,v_T\right)) \right)^{2} v_T^3},\label{bs-1''}
\end{eqnarray}
and $\left( BS^{-1}\right) ^{\prime \prime }\left(\hat{k}_t,\Lambda_\tau\right)A_\tau = 0$.

This gives us
\begin{eqnarray}
I\left( t,T,\tau,\hat{k}_t\right) &=& E_{t}\left[ v_{T}\right] 
+\frac12 E_t\bigg[\int_{t}^{\tau}\left( BS^{-1}\left( \hat{k}_t, \Lambda_r\right)\right) ^{\prime \prime\prime} \left(2 \int^\tau_r U_s D_r^W U_s ds\right) U_r dr\Bigg]\nonumber\\
&&+\frac{1}{2} E_t\Bigg[
\int_{t}^{\tau}\left( BS^{-1}\left( \hat{k}_t, \Lambda_r\right)\right) ^{(iv)} A_r U_r^2 dr \Bigg],
\label{I-E[v]}
\end{eqnarray}
and completes the proof.
\end{proof}

\begin{proof}[Proof of Theorem \ref{uncorrelated}]
The proof is decomposed in several steps.\\

\noindent{\it Step 1} Firstly, we show that
\begin{eqnarray}
\label{impliedexpansion}
I\left( t,T,\tau,\hat{k}_t\right) &=& E_{t}\left[ v_{T}\right] \nonumber\\
&&+\frac12\left( BS^{-1}\left( \hat{k}_t, \Lambda_t\right)\right) ^{\prime \prime\prime} E_t\Bigg[\int_{t}^{\tau}(D^-A)_r U_r dr\Bigg]\nonumber\\
&&+\frac14 \left( BS^{-1}\left( \hat{k}_t, \Lambda_t\right)\right) ^{(iv)} E_t\Bigg[
\int_{t}^{\tau} A_r U_r^2 dr \Bigg]\nonumber\\
&&+T_1+T_2+T_3+T_4,
\end{eqnarray}
where
\begin{eqnarray*}
T_1&=&E_t\Bigg[\int_{t}^{\tau}\left( BS^{-1}\left( \hat{k}_t, \Lambda_r\right)\right) ^{(iv)}(D^-\Psi)_r U_r dr\Bigg],
\\
T_2&=& \frac12 E_t\Bigg[\int_{t}^{\tau}\left( BS^{-1}\left( \hat{k}_t, \Lambda_r\right)\right) ^{(v)}\Psi_r U_r^2 dr\Bigg],
\\
T_3&=&E_t\Bigg[\int_{t}^{\tau}\left( BS^{-1}\left( \hat{k}_t, \Lambda_r\right)\right) ^{(v)}(D^-\Phi)_r U_r dr\Bigg],
\end{eqnarray*}
and
\begin{eqnarray*}
T_4&=& \frac12 E_t\Bigg[\int_{t}^{\tau}\left( BS^{-1}\left( \hat{k}_t, \Lambda_r\right)\right) ^{(vi)}\Phi_r U_r^2 dr\Bigg],
\end{eqnarray*}
with $\Psi_t:=\int_{t}^{\tau}(D^-A)_r U_r dr$ and $\Phi_t:=\int_{t}^{\tau} A_r U_r^2 dr$.
To this end, we apply the anticipating It\^o's formula to the processes
$$
\left( BS^{-1}\left( \hat{k}_t, \Lambda_t\right)\right) ^{\prime \prime\prime} \int_{t}^{\tau}(D^-A)_r U_r dr=:\left( BS^{-1}\left( \hat{k}_t, \Lambda_t\right)\right) ^{\prime \prime\prime}\Psi(t),
$$
and 
$$
\frac14 \left( BS^{-1}\left( \hat{k}_t, \Lambda_r\right)\right) ^{(iv)}
\int_{t}^{\tau} A_r U_r^2 dr =:\frac14 \left( BS^{-1}\left( \hat{k}_t, \Lambda_r\right)\right) ^{(iv)}\Phi(t).
$$
Then, the same arguments as in the proof of  Proposition \ref{General} give us
\begin{eqnarray}
\lefteqn{E_t\Bigg[\int_{t}^{\tau}\left( BS^{-1}\left( \hat{k}_t, \Lambda_r\right)\right) ^{\prime \prime\prime}(D^-A)_r U_r dr\Bigg]}\nonumber\\
&=&\left( BS^{-1}\left( \hat{k}_t, \Lambda_t\right)\right) ^{\prime \prime\prime} E_t\Bigg[\int_t^\tau(D^-A)_r U_r dr\Bigg]\nonumber\\
&&+ E_t\Bigg[\int_{t}^{\tau}\left( BS^{-1}\left( \hat{k}_t, \Lambda_r\right)\right) ^{(iv)}(D^-\Psi)_r U_r dr\Bigg]\nonumber\\
&&+ \frac12 E_t\Bigg[\int_{t}^{\tau}\left( BS^{-1}\left( \hat{k}_t, \Lambda_r\right)\right) ^{(v)}\Psi_r U_r^2 dr\Bigg]\nonumber\\
&=&\left( BS^{-1}\left( \hat{k}_t, \Lambda_t\right)\right) ^{\prime \prime\prime} E_t\Bigg[\int_t^\tau(D^-A)_r U_r dr\Bigg]+T_1+T_2,
\end{eqnarray}
and
\begin{eqnarray}
\lefteqn{E_t\Bigg[
\int_{t}^{\tau}\left( BS^{-1}\left( \hat{k}_t, \Lambda_r\right)\right) ^{(iv)} A_r U_r^2 dr \Bigg]}\nonumber\\
&=&\left( BS^{-1}\left( \hat{k}_t, \Lambda_t\right)\right) ^{(iv)} E_t\Bigg[
\int_{t}^{\tau}A_r U_r^2 dr \Bigg]\nonumber\\
&&+E_t\Bigg[
\int_{t}^{\tau}\left( BS^{-1}\left( \hat{k}_t, \Lambda_r\right)\right) ^{(v)} (D^-\Phi)_r U_r dr \Bigg]\nonumber\\
&&+\frac12 E_t\Bigg[\int_{t}^{\tau}\left( BS^{-1}\left( \hat{k}_t, \Lambda_r\right)\right) ^{(vi)}\Phi_r U_r^2 dr\Bigg]\nonumber\\
&=&\left( BS^{-1}\left( \hat{k}_t, \Lambda_t\right)\right) ^{(iv)} E_t\Bigg[
\int_{t}^{\tau}A_r U_r^2 dr \Bigg]+T_3+T_4.
\end{eqnarray}
{\it Step 2} Next, let us consider the term
$$
\left( BS^{-1}\left( \hat{k}_t, \Lambda_t\right)\right) ^{\prime \prime\prime} E_t\Bigg[\int_t^\tau (D^-A)_r U_r dr\Bigg].
$$
On one hand,
\begin{eqnarray}
\label{terceraderivada}
\lefteqn{\left( BS^{-1}\left( \hat{k}_t, \Lambda_t\right)\right) ^{\prime \prime\prime}}\nonumber\\
&=&\frac{-\frac{\partial^3 BS}{\partial \sigma^3}(\hat{k}_t, I(t,T,\tau,\hat{k}_t))\left(\frac{\partial BS}{\partial \sigma}(\hat{k}_t, I(t,T,\tau,\hat{k}_t))\right)^3
+3\left(\frac{\partial^2 BS}{\partial \sigma^2}(\hat{k}_t, I(t,T,\tau,\hat{k}_t))\right)^2\left(\frac{\partial BS}{\partial \sigma}(\hat{k}_t, I(t,T,\tau,\hat{k}_t))\right)^2
}{\left(\frac{\partial BS}{\partial \sigma}(\hat{k}_t, I(t,T,\tau,\hat{k}_t))\right)^7}\nonumber\\
&=&\frac{-\frac{\partial^3 BS}{\partial \sigma^3}(\hat{k}_t, I(t,T,\tau,\hat{k}_t))}{\left(\frac{\partial BS}{\partial \sigma}(\hat{k}_t, I(t,T,\tau,\hat{k}_t))\right)^4} 
+o\left((\tau-T)^{-\frac{1}{2}}\right)\nonumber\\
&=&(2\pi)^{\frac32}\exp\left(-3+\frac{3}{2}(I(t,T,\tau,\hat{k}_t))^2(\tau-T)\right)(\tau-T)^{-\frac12}+o\left((\tau-T)^{-\frac{1}{2}}\right).
\end{eqnarray}
On the other hand,
\begin{eqnarray}
\label{DA}
(D^-A)_r&=&\int_r^\tau D_r^WU_s^2ds=2\int_r^\tau U_sD_r^WU_sds.
\end{eqnarray}
The vega-delta-gamma relationship gives us
\begin{eqnarray}
\label{U}
U_s&=&E_s\left[\frac{\partial BS}{\partial \sigma}(T,\tau,0,\hat{k}_t,v_T)\frac{1}{2v_T(\tau-T)}\int_{s\vee T}^\tau D_s^W\sigma_u^2 du\right]  \nonumber\\
&=&\frac12    E_s\left[G(T,\tau,0,\hat{k}_t,v_T)\int_{s\vee T}^\tau D_s^W\sigma_u^2 du\right],
\end{eqnarray}
and
\begin{eqnarray}
\label{DU}
D_r^WU_s&=& E_s\Bigg[ \frac12 
G(T,\tau,0,\hat{k}_t,v_T)\left(\frac{d_1(\hat{k}_t,v_T)d_2(\hat{k}_t,v_T)}{2v_T(\tau-T)} - \frac{1}{2v_T^{2}(\tau-T)}\right)
\left(\int_{s\vee T}^\tau D_s^W\sigma_u^2 du\right)\left(\int_{r\vee T}^\tau D_r^W\sigma_u^2 du\right)
\nonumber\\
&&+\frac12 G(T,\tau,0,\hat{k}_t,v_T)\left(\int_{r\vee s\vee T}^\tau D_r^WD_s^W\sigma_u^2 du\right)\Bigg]\nonumber\\
\end{eqnarray}

Here, since (H2),
\begin{eqnarray}
\int^\tau_{s\vee T} D_s^W\sigma_u^2 du 
&\le& \int^\tau_{s\vee T} (\tau-s)^{H-\frac12}du\nonumber\\
&=& 1_{s>T}C(\tau-s)^{H+\frac12} + 1_{s<T} C(\tau - T)(\tau - s)^{H-\frac12},
\end{eqnarray}
where $C$s are different constants in each term, and
\begin{eqnarray}
\lefteqn{\int_r^\tau \int^\tau_{s\vee T} D_s^W\sigma_u^2 du \int^\tau_{s\vee T} D_s^W\sigma_u^2 du ds}\nonumber\\
&\le&
\int_r^\tau \int^\tau_{s\vee T} (\tau - s)^{H-\frac12} du \int^\tau_{s\vee T}  (\tau - s)^{H-\frac12} du ds
\nonumber\\
&=& 1_{r>T}C(\tau - r)^{2H+2} + 1_{r<T} C (\tau-T)^{2}(\tau - r)^{2H}
.\label{dv2dv2}
\end{eqnarray}

From the equation for $G$, the zero-vanna relationship $x-k = \frac12 I(t,T,\tau,\hat{k}_t) \sqrt{\tau-T}$ and (H2'), we deduce 

\begin{eqnarray}
\label{DA2}
(D^-A)_r
&=&
\frac12
\int_r^\tau E_s\left[\frac{e^{1}N'(d_2(\hat{k},v_T))}{v_T\sqrt{\tau-T}}\int_{s\vee T}^\tau D_s^W\sigma_u^2 du \right] \nonumber\\
&&\times E_s \Bigg[
\frac{e^{1}N'(d_2(\hat{k},v_T))}{v_T\sqrt{\tau-T}}
\frac{-1}{2 v_T^{2} (\tau-T)} 
\left(\int_{s\vee T}^\tau D_s^W\sigma_u^2 du\right)\left(\int_{r\vee T}^\tau D_r^W\sigma_u^2 du\right)
\Bigg] ds\nonumber\\
&&+
\frac12
\int_r^\tau E_s\left[\frac{e^{1}N'(d_2(\hat{k},v_T))}{v_T\sqrt{\tau-T}}\int_{s\vee T}^\tau D_s^W\sigma_u^2 du \right] \nonumber\\
&&\times E_s \Bigg[
\frac{e^{1}N'(d_2(\hat{k},v_T))}{v_T\sqrt{\tau-T}}
\frac{\frac14 (I(t,T,\tau,\hat{k}_t)^2 - v_t^2) (\tau-T)}{2 v_T^{2} (\tau-T)} 
\left(\int_{s\vee T}^\tau D_s^W\sigma_u^2 du\right)\left(\int_{r\vee T}^\tau D_r^W\sigma_u^2 du\right)
\Bigg] ds\nonumber\\
&&+
\frac12
\int_r^\tau E_s\left[\frac{e^{1}N'(d_2(\hat{k},v_T))}{v_T\sqrt{\tau-T}}\int_{s\vee T}^\tau D_s^W\sigma_u^2 du \right] \nonumber\\
&&\times E_s \Bigg[
\frac{e^{1}N'(d_2(\hat{k},v_T))}{v_T\sqrt{\tau-T}}
\left(\int_{r\vee s\vee T}^\tau D_r^WD_s^W\sigma_u^2 du\right)
\Bigg] ds\nonumber\\
&=& F_1+F_2+F_3.
\end{eqnarray}

Since \eqref{dv2dv2},
\begin{eqnarray}
F_1 
&=& O\left(\frac{C}{(\tau-T)^2}\int_r^{\tau} E_s\Bigg[ \int^\tau_{s\vee T} D_s^W\sigma_u^2 du\Bigg] E_s\Bigg[ \int^\tau_{s\vee T} D_s^W\sigma_u^2 du \int^\tau_{r\vee T} D_s^W\sigma_u^2 du\Bigg]ds\right) \nonumber\\
&=&
O\Bigg((\tau-T)^{-2}\Big\{1_{r>T}C(\tau - r)^{2H+2} + 1_{r<T} C (\tau-T)^{2}(\tau - r)^{2H}\Big\}
\Big\{ 1_{r>T}C(\tau-r)^{H+\frac12} + 1_{r<T} C(\tau - T)(\tau - r)^{H-\frac12} \Big\}\Bigg)
\nonumber\\
&=&
O\Bigg((\tau-T)^{-2}\Big\{1_{r>T}C(\tau-r)^{3H+\frac52} +  1_{r<T}(\tau - r)^{3H-\frac12}(\tau - T)^{3} \Big\}\Bigg)\nonumber\\
&=& O(\tau-T)^{-2}(\tau - r)^{3H-\frac12}(\tau - r\vee T)^2,\\
F_2&=&(\tau - T)F_1\\
F_3 
&=& O\left(\frac{C}{(\tau-T)}\int_r^{\tau} E_s\Bigg[ \int^\tau_{s\vee T} D_s^W\sigma_u^2 du\Bigg] E_s\Bigg[ \int^\tau_{s\vee T} D_r^W D_s^W\sigma_u^2 du \Bigg]ds\right) \nonumber\\
&=& O\Bigg(\frac{C}{(\tau-T)}\int_r^{\tau}  \left(1_{s>T}C(\tau-s)^{H+\frac12} + 1_{s<T} C(\tau - T)(\tau - s)^{H-\frac12}\right)\nonumber\\
&& \times \left(1_{s>T}(\tau-r)^{H-\frac12}(\tau - s)^{H+\frac12} + 1_{s<T}(\tau - T)(\tau - r)^{2H-1}\right) ds\Bigg)\nonumber\\
&=& O\Bigg((\tau-T)^{-1}\left(1_{r>T}(\tau-r)^{3H+\frac32} + 1_{r<T} (\tau-r)^{3H-\frac12}(\tau - T)^{2} \right)\Bigg)\nonumber\\
&=& O(\tau - T)^{-1}(\tau - r)^{3H-\frac12}(\tau - r\vee T)^2
\end{eqnarray}
Thus,
\begin{eqnarray}
(D^-A)_r &=&
O(\tau-T)^{-2}(\tau-r)^{3H-\frac12}(\tau - r\vee T)^3
+ O(\tau - T)^{-1}(\tau - r)^{3H-\frac12}(\tau - r\vee T)^2
\end{eqnarray}
and
\begin{eqnarray}
U_r &=& 
O\Bigg((\tau-T)^{-\frac12}\left(1_{r>T}C(\tau-r)^{H+\frac12} + 1_{r<T} C(\tau - T)(\tau - r)^{H-\frac12}\right)\Bigg)\nonumber\\
&=& O(\tau - T)^{-\frac12}(\tau - r)^{H-\frac12}(\tau - r\vee T)
\end{eqnarray}

Thus,
\begin{eqnarray}
\int^\tau_t (D^- A)_r U_r dr &=& 
O\left(\int^\tau_t (\tau-T)^{-\frac52}(\tau - r)^{4H-1}(\tau - r\vee T)^4 dr + \int^\tau_t (\tau-T)^{-\frac32}(\tau - r)^{4H-1}(\tau - r\vee T)^3 dr\right)\nonumber\\
&=& O(\tau-T)^{4H+\frac32} + O(\tau - T)^{\frac32}(\tau - t)^{4H} 
\label{dau}
\end{eqnarray}

which implies that
\begin{eqnarray}
\lefteqn{E_t\Bigg[\int_t^\tau (D^-A)_r U_r dr\Bigg]}\nonumber\\
&=& 
O(\tau-T)^{4H+\frac32} + O(\tau - T)^{\frac32}(\tau - t)^{4H} .
\end{eqnarray}
Jointly with (\ref{terceraderivada}) and (H2') gives
\begin{eqnarray}
\lefteqn{\left( BS^{-1}\left( \hat{k}_t, \Lambda_t\right)\right) ^{\prime \prime\prime} E_t\Bigg[\int_t^\tau (D^-A)_r U_r dr\Bigg]}\nonumber\\
&=&
O(\tau-T)^{4H+1} + O(\tau - T)(\tau - t)^{4H}.
\end{eqnarray}

{\it Step 3} We calculate the term
$$
\frac14 \left( BS^{-1}\left( \hat{k}_t, \Lambda_r\right)\right) ^{(iv)} E_t\Bigg[\int_{t}^{\tau} A_r U_r^2 dr \Bigg].
$$
Here,
\begin{equation}
\label{derivada4}
\left( BS^{-1}\left( \hat{k}_t, \Lambda_t\right)\right) ^{(iv)}
=-(2\pi)^{2}\exp\left(-4X_t+2(I(r,T,\tau,0,\hat{k}_t))^2(\tau-T)\right)(\tau-T)^{-1} + o\left((\tau-T)^{-1}\right).
\end{equation}
On the other hand, 
\begin{eqnarray}
\lefteqn{E_t\Bigg[\int_{t}^{\tau} A_r U_r^2 dr \Bigg]}\nonumber\\
&=&E_t\Bigg[\int_{t}^{T} \left(\int_r^\tau U_s^2ds\right) U_r^2 dr \Bigg]\nonumber\\
&=&\frac12E_t\Bigg[ \left(\int_t^\tau U_r^2dr\right)^2  \Bigg]\nonumber\\
&=&
\frac12E_t\Bigg[ \left(\int_t^\tau \left(E_r\left[\frac{\partial BS}{\partial \sigma}(t,T,\tau,0,\hat{k}_t,v_T)\frac{1}{2v_T(\tau - T)}\int_{r\vee T}^\tau D_r^W\sigma_s^2 dr\right]\right)^2ds\right)^2  \Bigg]\nonumber\\
&=& 
\left(\int^\tau_t (\tau- T)^{-1}(\tau- r)^{2H-1}(\tau - r\vee T)^{2}dr\right)^2
\nonumber\\
&=& 
O(\tau- T)^{4H+2} + O(\tau- T)^{2H+2}(\tau - t)^{2H} + O(\tau- T)^{2}(\tau - t)^{4H}.
%
\end{eqnarray}
Together with (\ref{derivada4}), this gives us 
$$
\left( BS^{-1}\left( \hat{k}_t, \Lambda_t\right)\right) ^{(iv)} E_t\Bigg[
\int_{t}^{\tau} A_r U_r^2 dr \Bigg]=
O(\tau- T)^{4H+1} + O(\tau- T)^{2H+1}(\tau - t)^{2H} + O(\tau- T)(\tau - t)^{4H}.
$$

{\it Step 4} Let us prove that 
$T_2+T_4=
O(\tau- T)^{4H+1}(\tau-t)^{2H} + O(\tau- T)^{2H+1}(\tau - t)^{4H} + O(\tau- T)(\tau - t)^{6H}
$.
The computations in Step 2 and Step 3 show that
$\Psi_r=
O(\tau-T)^{4H+\frac32} + O(\tau - T)^{\frac32}(\tau - t)^{4H}
$
and 
$\Phi_r=
O(\tau- T)^{4H+2} + O(\tau- T)^{2H+2}(\tau - t)^{2H} + O(\tau- T)^{2}(\tau - t)^{4H}
$.

In addition, 
$U_r= O(\tau - T)^{-\frac12}(\tau - r)^{H-\frac12}(\tau - r\vee T)$
and some computations shows
$$
BS^{-1}\left( \hat{k}_t, \Lambda_r\right)^{(v)}\leq C(\tau-T)^{-\frac32},
$$
and 
$$
BS^{-1}\left( \hat{k}_t, \Lambda_r\right)^{(vi)}\leq C(\tau-T)^{-2},
$$
for some positive constant $C$. Then, direct computations gives us 
$T_2+T_4=
O(\tau- T)^{4H+1}(\tau-t)^{2H} + O(\tau- T)^{2H+1}(\tau - t)^{4H} + O(\tau- T)(\tau - t)^{6H}
$. 

{\it Step 5} Finally, we show $T_1+T_3=O(\tau-t)^{6H+1} + O(\tau - T)^{4H+1}(\tau-t)^{2H} + O(\tau - T)(\tau - t)^{6H}
$. 
Here,
\begin{eqnarray}
D^-\Psi_t&:=&\int_{t}^{\tau}D_t^W ((D^-A)_r U_r )dr\nonumber\\
&=&\int_{t}^{\tau}(D_t^W (D^-A)_r) U_r dr+\int_{t}^{T} (D^-A)_r D_t^W U_r dr,
\end{eqnarray}
 and
\begin{eqnarray}
D^-\Phi_t&:=&\int_{t}^{\tau} D_t^W(A_r U_r^2)dr\nonumber\\
&=&\int_{t}^{\tau}(D_t^W A_r) U_r^2 dr+2\int_{t}^{\tau} U_r A_r D_t^W( U_r) dr,
\end{eqnarray}
where
\begin{eqnarray}
U_s&=& \frac12 E_s\left[G(T,\tau,0,\hat{k}_t,v_T)\int_{s\vee T}^\tau D_s^W\sigma_u^2 du\right]\nonumber\\
&=& 
O(\tau-s\vee T)(\tau-s)^{H-\frac12}(\tau-T)^{-\frac12}\\
U_s^2
&=& 
O(\tau-s\vee T)^2(\tau-s)^{2H-1}(\tau-T)^{-1}
\\
D_t^WU_s&=& \frac12 E_s\Bigg[ 
G(T,\tau,0,\hat{k}_t,v_T)\left(\frac{d_1(\hat{k}_t,v_T)d_2(\hat{k}_t,v_T)}{2v_T^2(\tau-T)} - \frac{1}{2v_T^{2}(\tau-T)}\right)
\left(\int_{s\vee T}^\tau D_s^W\sigma_u^2 du\right)\left(\int_{T}^\tau D_t^W\sigma_u^2 du\right)
\nonumber\\
&&+ G(T,\tau,0,\hat{k}_t,v_T)\left(\int_{s\vee T}^\tau D_t^WD_s^W\sigma_u^2 du\right)\Bigg]\nonumber\\
&=&
O(\tau-s\vee T)(\tau-s)^{H-\frac12}(\tau-t)^{H-\frac12}(\tau-T)^{-\frac12}
\\
D_t^WD_r^WU_s
&=& \frac18 E_s\Bigg[
G(T,\tau,0,\hat{k}_t,v_T)\frac{1}{8v_t^4(\tau-T)^2}
\left(-d_1d_2(1-d_1d_2)-d_1^2-d_2^2 - d_1d_2 + \frac12\right)\nonumber\\
&&\times
\left(\int_{s\vee T}^\tau D_s^W\sigma_u^2 du\right)\left(\int_{r\vee T}^\tau D_t^W\sigma_u^2 du\right)\left(\int_{t\vee T}^\tau D_t^W\sigma_u^2 du\right)\nonumber\\
&&+
2fG(T,\tau,0,\hat{k}_t,v_T)\left(\frac{d_1(\hat{k}_t,v_T)d_2(\hat{k}_t,v_T)}{2v_T^2(\tau-T)} - \frac{1}{2v_T^{2}(\tau-T)}\right)\nonumber\\
&&\times \left(
\left(\int_{s\vee T}^\tau D_t^WD_s^W\sigma_u^2 du\right)\left(\int_{r\vee T}^\tau D_r^W\sigma_u^2 du\right)
+\left(\int_{s\vee T}^\tau D_s^W\sigma_u^2 du\right)\left(\int_{r\vee T}^\tau D_t^WD_r^W\sigma_u^2 du\right)
\right)\nonumber\\
&&+ G(T,\tau,0,\hat{k}_t,v_T)\left(\int_{r\vee s\vee T}^\tau D_t^WD_r^WD_s^W\sigma_u^2 du\right)\Bigg]\nonumber\\
&=&O(\tau-s)^{H-\frac12}(\tau-r)^{H-\frac12}(\tau-t)^{H-\frac12}(\tau-T)^{-\frac32}(\tau-s\vee T)(\tau-r\vee T)\nonumber\\
&&+O(\tau-s)^{H-\frac12}(\tau-r)^{H-\frac12}(\tau-t)^{H-\frac12}(\tau-T)^{-\frac12}(\tau-r\vee s\vee T)\end{eqnarray}
and under (H2'), we get
\begin{eqnarray}
\label{DA}
D_t^W (D^-A)_r&=&\int_r^\tau D_t^W(U_s D_r^WU_s)ds\nonumber\\
&=&2\int_r^\tau D_t^WU_s D_r^WU_s ds + \int_r^\tau U_s(D_t^W D_r^WU_s )ds
\end{eqnarray}
Then, 
\begin{eqnarray}
U_s D_t^WD_r^WU_s
&=& O(\tau-s)^{2H-1}(\tau-r)^{H-\frac12}(\tau-t)^{H-\frac12}(\tau-T)^{-2}(\tau-s\vee T)^2(\tau-r\vee T)\nonumber\\
&&+O(\tau-s)^{2H-1}(\tau-r)^{H-\frac12}(\tau-t)^{H-\frac12}(\tau-T)^{-1}(\tau-s\vee T)\nonumber\\
&&\times(\tau-s\vee r\vee T)\\
\int^\tau_r U_s D_t^WD_r^WU_s ds
&=& O(\tau-r)^{3H-\frac12}(\tau-t)^{H-\frac12}(\tau-T)^{-2}(\tau-r\vee T)^3\nonumber\\
&&+O(\tau-r)^{3H-\frac12}(\tau-t)^{H-\frac12}(\tau-T)^{-1}(\tau-r\vee T)^2\\
D_t^WU_sD_r^WU_s
&=& O(\tau-s)^{2H-1}(\tau-r)^{H-\frac12}(\tau-t)^{H-\frac12}(\tau-T)^{-1}(\tau-s\vee T)(\tau-r\vee T)
\\
\int^\tau_r D_t^WU_s D_r^WU_s ds
&=& O(\tau-r)^{3H-\frac12}(\tau-t)^{H-\frac12}(\tau-T)^{-1}(\tau-r\vee T)^2
\end{eqnarray}
and 
\begin{eqnarray}
D^-\Psi
&=& \int^\tau_t D_t^W(D^-A)_rU_r dr + \int^\tau_t (D^-A)_r D_t^WU_r dr\nonumber\\
&=& O\left(\int^\tau_t (\tau - r)^{4H-1}(\tau - t)^{H-\frac12}(\tau - T)^{-\frac52}(\tau - r\vee T)^4 dr\right)\nonumber\\
&& + O\left(\int^\tau_t (\tau - r)^{4H-1}(\tau - t)^{H-\frac12}(\tau - T)^{-\frac32}(\tau - r\vee T)^3 dr \right)\nonumber\\
&=& O(\tau - T)^{4H+\frac32}(\tau - t)^{H-\frac12} + O(\tau - T)^{\frac32}(\tau - t)^{5H-\frac12}.
\end{eqnarray}
Also,
\begin{eqnarray}
D^-\Phi_t
&=&\int_{t}^{\tau} D_t^W(A_r U_r^2)dr\nonumber\\
&=&\int_{t}^{\tau}(D_t^W A_r) U_r^2 dr+2\int_{t}^{\tau} U_r A_r D_t^W( U_r) dr
\end{eqnarray}
and
\begin{eqnarray}
D_t^WA_r
&=&
\int^\tau_r U_s D_t^W U_s ds\nonumber\\
&=&
O\left(\int^\tau_r (\tau - T)^{-1}(\tau - s)^{2H-1}(\tau-t)^{H-\frac12}(\tau - s\vee T)^{2}ds\right)\nonumber\\
&=&O(\tau - T)^{-1}(\tau - r)^{2H}(\tau-t)^{H-\frac12}(\tau - r\vee T)^{4}
\\
(D_t^WA_r) U_r^2
&=&
O(\tau - T)^{-2}(\tau - r)^{4H-1}(\tau - t)^{H-\frac12}(\tau - r\vee T)^{4} \nonumber\\
&&+ O(\tau - T)^{-2}(\tau - r)^{5H-\frac32}(\tau - r\vee T)^{4}\\
U_r A_r D_t^WU_r 
&=&
O(\tau - T)^{-2}(\tau - r)^{4H-1}(\tau - t)^{H-\frac12}(\tau - r\vee T)^{-2}\\
\int^\tau_t (D_t^WA_r) U_r^2 dr
&=&
O(\tau - T)^{4H+2}(\tau - t)^{H-\frac12}+ O(\tau - T)^{2}(\tau - t)^{5H-\frac12}\\
\int^\tau_t U_r A_r D_t^WU_r dr
&=&
O(\tau - T)^{4H+2}(\tau - t)^{H-\frac12}+ O(\tau - T)^{2}(\tau - t)^{5H-\frac12}
\end{eqnarray}
Thus,
\begin{eqnarray}
D^-\Phi_t
&=&
O(\tau - T)^{4H+2}(\tau - t)^{H-\frac12}+ O(\tau - T)^{2}(\tau - t)^{5H-\frac12}
\end{eqnarray}

Again, for some positive constant $C$, direct computations give us
$$
BS^{-1}\left( \hat{k}_t, \Lambda_r\right)^{(iv)}\leq C(\tau-T)^{-1},
$$
and 
$$
BS^{-1}\left( \hat{k}_t, \Lambda_r\right)^{(v)}\leq C(T-r)^{-\frac32},
$$
Then, we get
$T_1+T_3=
O(\tau-t)^{6H+1} + O(\tau - T)^{4H+1}(\tau-t)^{2H} + O(\tau - T)(\tau - t)^{6H}
$.
\end{proof}

\begin{proof}[Proof of Theorem \ref{themaintheorem}]
The proof follows similar arguments as in the proof of Theorem 4.2 in Al\`os and Shiraya (2019). Notice that Proposition \ref{Theoremcorrelatedcase} gives us that
$$
I( t,T,\tau,\hat{k}_t)  -E_t[v_T]=T_1+T_2,
$$
where
\begin{eqnarray*}
T_1&=&I^{0}( t,T,\tau,0,\hat{k}_t) -E_t[v_T],\\
T_2&=&\frac{\rho }{2}E_t\left[\int_{T}^{\tau}( BS^{-1}) ^{\prime }( \hat{k}_t,\Gamma _{s}) H(s,\tau,X_{s}-X_T,\hat{k}_t,v_{s})\zeta_{s}ds\right].
\end{eqnarray*}
We have seen that in Theorem \ref{uncorrelated} that, if $H<\frac12$, $T_1=o( (\tau-T)^{2H})$. Now, let us study $T_2$. 
Towards this end, we apply the anticipating It\^{o}'s formula (\ref{aito}) to the process
\[
H(s,\tau,X_{s}-X_T,\hat{k}_t,v_{s})J_{s},
\]%
where $J_{s}=\int_{s}^{\tau}( BS^{-1}) ^{\prime }( \hat{k}_t,\Gamma _{u})\zeta_{u}du$. Then,
taking conditional expectations we get
\begin{eqnarray*}
0 &=&E_{t}\Bigg[ H(T,\tau,0,\hat{k}_t,v_{t})J_{t}  \\
&&+\int^\tau_T H(s,\tau,X_{s}-X_T,\hat{k}_t,v_{s}) dJ_{s}\\
&&+\int^\tau_T\frac{\partial^2}{\partial x \partial \sigma} H(s,\tau,X_{s}-X_T,\hat{k}_t,v_{s}) J_{s} \frac{\partial v}{\partial y} (D^W_s Y_s) \sigma_s ds
\\
&&+\int^\tau_T \frac{\partial}{\partial x} H(s,\tau,X_{s}-X_T,\hat{k}_t,v_{s}) (D^W_s J_s) \sigma_s ds \\
&&+\int_{T}^{\tau} \frac{\partial}{\partial t} H(s,\tau,X_{s}-X_T,\hat{k}_t,v_{s})J_{s} ds \\
&&+\int_{T}^{\tau}\frac{\partial}{\partial \sigma} H(s,\tau,X_{s}-X_T,\hat{k}_t,v_{s})\frac{\partial v}{\partial t} J_{s}ds  \\
&&+\int_{T}^{\tau} \frac{\partial}{\partial \sigma} H(s,\tau,X_{s}-X_T,\hat{k}_t,v_{s})\frac{\partial v}{\partial y} J_{s}dY_s \\
&&+\int_{T}^{\tau}\frac{\partial}{\partial x} H(s,\tau,X_{s}-X_T,\hat{k}_t,v_{s}) J_{s} dX_s \\
&&+\frac{1}{2}\int_{T}^{\tau} \frac{\partial^2}{\partial x^2} H(s,\tau,X_{s}-X_T,\hat{k}_t,v_{s}) J_{s} d\langle X\rangle_s \Bigg]. 
\end{eqnarray*}
Now, using  the relationships
\begin{eqnarray*}
&&\frac{1}{\sigma(\tau-s)}\frac{\partial}{\partial \sigma}BS(s,\tau,x,k,\sigma)=\left(\frac{\partial^2}{\partial x^2} - \frac{\partial}{\partial x}\right)BS(s,\tau,x,k,\sigma),\\
&&\left(\frac{\partial}{\partial t} + \frac{1}{2}\sigma^2\frac{\partial^2}{\partial x^2}  - \frac{1}{2}\sigma^2 \frac{\partial}{\partial x} \right)BS(\tau,s,x,k,\sigma)=0,\\
&&D^W_s J_s = \rho \int_{s}^{\tau}( BS^{-1}) ^{\prime }( \hat{k}_t,\Gamma _{r})D^W_s \zeta_{r}dr,\\
&&D^W_s Y_s = \rho \int^\tau_s D^W_s \sigma^2_r dr,
\end{eqnarray*}
we obtain
\begin{eqnarray*}
0&=&E_{t}\Bigg[ H(T,\tau,0,\hat{k}_t,v_{t})J_{t}  \\
&&-\int_{T}^{\tau}H(s,\tau,X_{s}-X_T,\hat{k}_t,v_{s})( BS^{-1}) ^{\prime
}( X_t,\Gamma _{s}) \zeta _{s}ds \\
&& 
+ \frac{\rho}{2} \int_{T}^{\tau}\left(\frac{\partial^3}{\partial x^3} - \frac{\partial^2}{\partial x^2} \right) H(s,\tau,X_{s}-X_T,\hat{k}_t,v_{s})J_{s}\zeta_s ds
\\
&&+ \rho \int_{T}^{\tau}\frac{\partial }{\partial x}H(s,\tau,X_{s}-X_T,\hat{k}_t,v_{s})\left( \int_{s}^{\tau}( BS^{-1}) ^{\prime }( \hat{k}_t,\Gamma _{r}) ( D_{s}^{W}\zeta _{r}) dr\right) \sigma _{s}ds
\Bigg],
\end{eqnarray*}
which implies that
\begin{eqnarray*}
T_{2} &=&E_{t}\Bigg[ \frac{\rho }{2} H(T,\tau,0,\hat{k}_t,v_{t})J_{t}\\
&&
+ \frac{\rho^2}{4}\int_{T}^{\tau} \left(\frac{\partial^3}{\partial x^3} - \frac{\partial^2}{\partial x^2} \right) H(s,\tau,X_{s}-X_T,\hat{k}_t,v_{s})J_{s}\zeta_s ds
\\
&&
+ \frac{\rho^2}{2}\int_{T}^{\tau}\frac{\partial }{\partial x}H(s,\tau,X_{s}-X_T,\hat{k}_t,v_{s})\left( \int_{s}^{\tau}( BS^{-1}) ^{\prime }(\hat{k}_t,\Gamma _{r}) ( D_{s}^{W}\zeta _{r}) dr\right) \sigma _{s}ds\Bigg]\\
&=&T_{2}^{1}+T_{2}^{2}+T_{2}^{3}.
\end{eqnarray*}
Now, the study of $T_2$ is decomposed into two steps.

\textit{Step 1 } As
\begin{eqnarray*}
H(T,\tau,0,\hat{k}_t,v_{t})
&=&\frac{e^{X_{t}}N^{\prime }(d_1(
\hat{k}_t,v_t) )}{v_t\sqrt{\tau-T}}\left( 1-\frac{d_1( \hat{k}_t,v_t) 
}{v_t\sqrt{\tau-T}}\right)\\
&=&\frac{e^{X_{t}}N^{\prime }(d_1(
\hat{k}_t,v_t) )}{2v_t^3}\left( (I_t(t,T,\tau, \hat{k}_t))^2-v_t^2\right).
\end{eqnarray*}
we have that
\begin{eqnarray}
\label{T21}
\lefteqn{\lim_{\tau\rightarrow T}\frac{T_{2}^{1}}{(\tau-T)^{2H}}}\nonumber\\
&\le&\lim_{T\to t}\frac{\rho}{2(\tau-T)^{2H}} E_{t}\Bigg[ \frac{e^{X_{t}}N^{\prime }(d_1(
\hat{k}_t,v_t) )}{2v_t^3\sqrt{\tau-T}}\left( (I_t(t,T,\tau,\hat{k}_t))^2-v_t^2\right)\nonumber\\
&&\times\int_{T}^{\tau}\frac{1}{e^{X_{t}}N^{\prime }(d_{+}\left( \hat{k}_t,BS^{-1}(\hat{k}_t,\Gamma _{s})\right)) \sqrt{\tau-T}}\zeta _{s}ds \Bigg].
\end{eqnarray}%
Then, a direct computation give us that $\lim_{\tau\rightarrow T}\frac{T_{2}^{1}}{(\tau-T)^{2H}}=0$.

\textit{Step 2}. In order to see that $T_2^2$ and $T_2^3$ are $O(\tau-T)^{2H}$ we apply again the anticipating It\^{o}'s formula to the processes%
\begin{eqnarray*}
\left( \frac{\partial ^{3}}{\partial x^{3}}-%
\frac{\partial^2 }{\partial x^2}\right) H(s,\tau,X_{s}-X_T,\hat{k}_t,v_{s})Z_{s},
\end{eqnarray*}
and%
\[
\frac{\partial H}{\partial x}(s,\tau,X_{s} -X_T,\hat{k}_t,v_{s})R_{s},
\]%
where 
\begin{eqnarray*}
Z_{s}&:=&\int_{s}^{\tau} \zeta_u J_{u} du,\\
R_{s}&:=&\int_{s}^{\tau} \left(\int_{u}^{\tau}(BS^{-1}) ^{\prime }(  \hat{k}_t,\Gamma _{r}) (D_{s}^{W}\zeta_{r}) dr\right) \sigma _{u}du.
\end{eqnarray*}
Then we get%
\begin{eqnarray}
T_{2}^{2}
&=&\frac{\rho^2}{4}E_{t}\Bigg[ \left( \frac{\partial ^{3}}{\partial x^{3}}-\frac{\partial^2 }{\partial x^2}\right) 
H(T,\tau,0, \hat{k}_t,v_{t})Z_t
\nonumber \\
&&+ \frac{\rho}{2}\int_{T}^{\tau}\left( \frac{\partial^{3}}{\partial x^{3}}-
\frac{\partial^2 }{\partial x^2}\right)^2 H(s,\tau,X_{s}-X_T,\hat{k}_t,v_{s}) Z_{s}\zeta_s ds  \nonumber \\
&&+\rho \int_{T}^{\tau}\frac{\partial}{\partial x}\left( \frac{\partial ^{3}}{\partial x^{3}}-
\frac{\partial^2 }{\partial x^2}\right) 
H(s,\tau,X_{s}-X_T,\hat{k}_t,v_{s})(D_{s}^{W} Z_s ) \sigma_s ds\Bigg], \label{T22}
\end{eqnarray}%
and 
\begin{eqnarray}
T_{2}^{3} &=&
\frac{\rho^2}{2}E_{t}\Bigg[\frac{\partial H}{\partial x}%
(T,\tau,0,\hat{k}_t,v_{t})R_t \nonumber \\
&&+\frac{\rho }{2}\int_{T}^{\tau}\left( \frac{\partial ^{3}}{\partial x^{3}}-%
\frac{\partial^2 }{\partial x^2}\right) \frac{\partial H}{\partial x}%
(s,\tau,X_{s},\hat{k}_t,v_{s}) R_{s} \zeta_{s}ds  \nonumber \\
&&+\rho \int_{T}^{\tau}\frac{\partial ^{2}H}{\partial x^{2}}(s,\tau,X_{s},\hat{k}_t,v_{s})\nonumber\\
&&\hspace{0.5cm}\times\left(\int_{s}^{\tau}\int_{r}^{\tau}\left( 
BS^{-1}\right) ^{\prime }(  \hat{k}_t,\Gamma_{u}) (
D_{s}^{W}D_{r}^{W}\zeta_{u}) dudr\right) \sigma_{s} ds\Bigg].  \label{T23}
\end{eqnarray}%
Lemma 4.1 in Al\`{o}s, Le\'{o}n and Vives (2007) gives us that the last two terms
in (\ref{T22}) and (\ref{T23}) are $O(\nu^3(\tau-T)^{3H})$. 
Now, as

\begin{eqnarray*}
\lefteqn{\left|\left( \frac{\partial ^{3}}{\partial x^{3}}-\frac{\partial^2 }{\partial x^2}\right) H(s,\tau,0,\hat{k}_t,v_{t})\right|}\\
&=& 
\Bigg|- \frac{d_1\left( \hat{k}_t,v_t\right)}{v_t \sqrt{\tau-s}} 
\frac{e^{X_{t}}N^{\prime }(d_1\left(\hat{k}_t,v_t\right) )}{v_t\sqrt{\tau-s}}
\left( 1-\frac{d_1\left( \hat{k}_t,v_t\right) }{v_t\sqrt{\tau-s}}\right)^3
\nonumber\\
&&- \left(3 - \frac{d_1\left( \hat{k}_t,v_t\right)}{v_t\sqrt{\tau-s}} \right)
\frac{e^{X_{t}}N^{\prime }(d_1\left(\hat{k}_t,v_t\right) )}{(v_t\sqrt{\tau-s})^3}
\left( 1-\frac{d_1\left( \hat{k}_t,v_t\right) }{v_t\sqrt{\tau-s}}\right)
+ 3\frac{e^{X_{t}}N^{\prime }(d_1\left(\hat{k}_t,v_t\right) )}{(v_t\sqrt{\tau-Ts})^5}\Bigg|\\
&=& 
\frac{3e^{X_{t}}N^{\prime }(d_1\left(\hat{k}_t,v_t\right) )}{v_t^5}(\tau-s)^{-\frac{5}{2}}
+O(\tau-T)^{-\frac{3}{2}},
\end{eqnarray*}
and
\begin{eqnarray*}
\lefteqn{\left|\frac{\partial H}{\partial x}(T,\tau,X_t,\hat{k}_t,v_{t})\right|}\\
&=&\left|\frac{e^{X_{t}}N^{\prime }(d_{1}\left(
\hat{k}_t,v_t\right) )}{v_t\sqrt{\tau-T}}\left( 1-\frac{d_{1}\left( \hat{k}_t,v_t\right) 
}{v_t\sqrt{\tau-T}}\right)^2 
- \frac{e^{X_{t}}N^{\prime }(d_{1}\left(\hat{k}_t,v_t\right) )}{(v_t\sqrt{\tau-T})^3}\right|\\
&=&
\frac{e^{X_{t}}N^{\prime }(d_1\left(\hat{k}_t,v_t\right) )}{v_t^3}(\tau-T)^{-\frac{3}{2}}
+O((\tau-T)^{-\frac{1}{2}}.
\end{eqnarray*}
\begin{eqnarray}
\label{LT22}
\lefteqn{\lim_{\tau\to T}{T_{2}^{2}} }\nonumber\\
&=& \frac{\rho^2}{4}E_{t}\left[ \left( \frac{\partial ^{3}}{\partial x^{3}}-\frac{\partial^2 }{\partial x^2}\right)H(T,\tau,0,\hat{k}_t,v_{t}) Z_t\right]\nonumber\\
&=&\frac{\rho^2}{4(\tau-T)^3}E_{t}\Bigg[ 3\frac{e^{X_{t}}N'(d_1(\hat{k}_t,v_t)) }{v_{t}^{5}}\nonumber\\
&& \times \int_{T}^{\tau} \sigma_s \left(\int_{s}^{\tau}D_s^{W}\sigma _{r}^{2}dr \right) \left(\int_{s}^{\tau}\frac{\zeta_{r}}{e^{X_{t}}N^{\prime }(d_1( \hat{k}_t,BS^{-1}(\hat{k}_t,\Gamma _{r}))) }dr\right) ds \Bigg]\nonumber\\
&=&\lim_{\tau\to T}\frac{3\rho ^{2}}{4\sigma_{t}^{5}(\tau-T)^{3}}E_{t}\left[ \int_{T}^{\tau}\left( \int_{s}^{\tau}D_{s}^{W}\sigma _{r}^{2}dr\right) \left(\int_{s}^{\tau}\zeta _{r}dr\right) \sigma _{s}ds\right]\nonumber\\
&=&\lim_{\tau\to T}\frac{3\rho ^{2}}{4\sigma _{t}^{5}(\tau-T)^{3}}
E_{t}\left[\int_{T}^{\tau}\left( \int_{s}^{\tau}D_{s}^{W}\sigma _{r}^{2}dr\right) \left(\int_{s}^{\tau}\sigma _{r} \int_{r}^{\tau}D_{r}^{W}\sigma _{\theta}^{2}d\theta  dr\right) \sigma _{s}ds \right] \nonumber\\
&=&\lim_{\tau\to T}\frac{3\rho ^{2}}{4\sigma _{t}^{3}(\tau-T)^{3}}E_{t}\left[\int_{T}^{\tau}\left(\int_{s}^{\tau}D_{s}^{W}\sigma _{r}^{2}dr\right) \left( \int_{s}^{\tau}\int_{r}^{\tau}D_{r}^{W}\sigma _{\theta }^{2}d\theta dr\right) ds\right]\nonumber\\
&=&\lim_{\tau\to T}\frac{3\rho ^{2}}{8\sigma _{t}^{3}(\tau-T)^{3}}E_{t}\left[\left( \int_{T}^{\tau}\int_{s}^{\tau}D_{s}^{W}\sigma _{r}^{2} drds\right) ^{2}\right], 
\end{eqnarray}%
and 
\begin{eqnarray}
\label{TL23}
\lefteqn{\lim_{\tau\to T}{T_{2}^{3}}}\nonumber\\
&=&\lim_{\tau\to T}\frac{\rho^2}{{2}}E_{t}\left[\frac{\partial H}{\partial x}(T,\tau,X_t,\hat{k}_t,v_{t})R_t\right]\nonumber\\
&=&\lim_{\tau\to T}\frac{\rho^2}{{2}}E_{t}\Bigg[
\frac{1}{4}\frac{e^{X_{t}}N^{\prime }(d_1( \hat{k}_t,v_t) )}{( v_t \sqrt{\tau-T})^{3}}\left( v_t^{2}(\tau-T)-4\right)\nonumber\\
&&\times\int_{T}^{\tau} \int_{s}^{\tau}\frac{1}{e^{X_{r}}N^{\prime }(d_1(  \hat{k}_t,BS^{-1}(\hat{k}_t,\Gamma _{r}))) \sqrt{\tau-T}} \left(D_{s}^{W}\left(\sigma_r \int_{r}^{\tau}D_s^{W}\sigma _{u}^{2}du\right)\right) dr \sigma _{s}ds \Bigg]\nonumber\\
&=&-\lim_{\tau\to T}\frac{\rho^{2}}{2\sigma _{t}^{2}(\tau-T)^{2}}
E_{t}\Bigg[ \int_{T}^{\tau} \int_{s}^{\tau}D_{s}^{W}\sigma_{r} \int_{r}^{\tau}D_{r}^{W}\sigma _{u}^{2}du dr ds \nonumber\\
&& +\int_{T}^{\tau} \int_{s}^{\tau} \sigma_{r}\int_{r}^{\tau}D_{s}^{W}D_{r}^{W}\sigma _{u}^{2}du dr ds \Bigg]\nonumber \\
&=&-\lim_{\tau\to T}\frac{\rho ^{2}}{{2\sigma _{t}(\tau-T)^2}}E_{t}\left[ \int_{T}^{\tau}\left( \int_{s}^{\tau} D_{s}^{W}\sigma_{r} dr\right) ^{2}ds\right] \nonumber\\
&&-\lim_{\tau\to T}\frac{\rho^{2}}{2\sigma_{t}(\tau-T)^{2}}E_{t}\left[ \int_{T}^{\tau} \int_{s}^{\tau} \int_{r}^{\tau}D_{s}^{W}D_{r}^{W}\sigma _{u}^{2}du dr ds\right],
\end{eqnarray}
Let us now summarize the previous computations. We have seen that
\begin{eqnarray}
I(t,T,\tau, X_t,\hat{k}_t) -E_t[v_t]&=&T_1+T_2\nonumber\\
&=&T_1+T_2^{1,1}+T_2^{1,2}+T_2^2+T_2^3
\end{eqnarray}
where $$T_1+T_2^{1,2}=o(\tau-T)^{2H},$$
\begin{eqnarray*}
T_2^{1,1}&=&(I(t,T,\tau,X_t,\hat{k}_t) - E_t[v_t])\frac{\rho}{4\sigma_t^2(\tau-T)}
E_t \left[ (I(t,T,\tau, X_t,\hat{k}_t)+v_t)\int_{T}^{\tau} \int_{s}^{\tau}D_{s}^{W}\sigma _{r}^{2}dr ds\right],\nonumber\\
T_2^2&=&\frac{3\rho ^{2}}{8\sigma _{t}^{3}(\tau-T)^{3}}E_{t}\left[\left( \int_{T}^{\tau}\int_{s}^{\tau}D_{s}^{W}\sigma _{r}^{2} drds\right) ^{2}\right]+o(\tau-T)^{2H},
\end{eqnarray*}
and
\begin{eqnarray}
T_2^3&=&-\lim_{\tau\to T}\frac{\rho ^{2}}{{2\sigma _{t}(\tau-T)^2 }}E_{t}\left[ \int_{T}^{\tau}\left( \int_{s}^{\tau} D_{s}^{W}\sigma_{r} dr\right) ^{2}ds\right] \nonumber\\
&&-\lim_{\tau\to T}\frac{\rho^{2}}{2\sigma_{t}(\tau-T)^{2 }}E_{t}\left[ \int_{T}^{\tau} \int_{s}^{\tau} \int_{r}^{\tau}D_{s}^{W}D_{r}^{W}\sigma _{u}^{2}du dr ds\right]\nonumber\\
&&+o(\tau-T)^{2H}.
\end{eqnarray}
Then, as there is some $\epsilon$ such that, if $\tau-T<\epsilon$ 
$$
\left|\frac{\rho}{4\sigma_t^2(\tau-T)}
E_t \left( (I_t(t,T,\tau,\hat{k}_t)+v_t)\int_{T}^{\tau} \int_{s}^{\tau}D_{s}^{W}\sigma _{r}^{2}dr ds\right)\right|<1
$$
we can write
\begin{eqnarray}
\lim_{\tau\to T} \frac{I(t,T,\tau, X_t,\hat{k}_t) -E_t[v_t]}{(\tau-T)^{2H}}
&=&\lim_{\tau\to T}\frac{1}{(\tau-T)^{2H}}\frac{T_1+T_2^{1,2}+T_2^2+T_3^3}{1-\frac{\rho}{4\sigma_t^2(\tau-T)}
E_t \left( (I_t(t,T,\tau,\hat{k}_t)+v_t)\int_{T}^{\tau} \int_{s}^{\tau} D_{s}^{W}\sigma _{r}^{2}dr ds\right)}\nonumber\\
&=&\lim_{\tau\to T}\frac{3\rho ^{2}}{8\sigma _{t}^{3}(\tau-T)^{3+2H}}E_{t}\left[\left( \int_{T}^{\tau} \int_{s}^{\tau} D_{s}^{W}\sigma _{r}^{2} drds\right) ^{2}\right]\nonumber\\
&&-\lim_{\tau\to T}\frac{\rho^{2}}{2\sigma_{t}(\tau-T)^{2+2H }}E_{t}\left[ \int_{T}^{\tau} \int_{s}^{\tau} \int_{r}^{\tau}D_{s}^{W}D_{r}^{W}\sigma _{u}^{2}du dr ds\right]\nonumber\\
&&-\lim_{\tau\to T}\frac{\rho^{2}}{{2\sigma_{t}(\tau-T)^{2+2H }}}E_{t}\left[ \int_{T}^{\tau} \int_{s}^{\tau} \int_{r}^{\tau}D_{s}^{W}D_{r}^{W}\sigma _{u}^{2}du dr ds\right],
\end{eqnarray}
as we wanted to prove.
\end{proof}

\end{document}